\documentclass[fleqn,english,twocolumn,secthm,amsthm]{autart}
\usepackage{graphicx}      
\usepackage{bm}
\usepackage{amsmath}
\usepackage{amssymb}
\usepackage{color}
\usepackage{mathtools}
\usepackage{caption}
\usepackage{multirow}
\usepackage{subfig}
\usepackage{tikz}
\usetikzlibrary{calc}
\usepackage{optidef}
\usepackage{mathtools, cuted}
\usepackage{lipsum}  

\usepackage{hyphenat}
\newcommand{\real}{\mathbb{R}}
\newcommand{\nonnegativeInt}{\mathbb{N}}
\newcommand{\proba}{\mathbb{P}}
\newcommand{\dimsys}{\nu}
\newcommand{\coeffindices}{\{0,1,\ldots,\dimsys\}}
\newcommand{\dimtru}{N_{\mathrm{T}}}
\newcommand{\coordx}{p_{\mathrm{x}}} 
\newcommand{\coordy}{p_{\mathrm{y}}} 
\newcommand{\dotcoordx}{\dot{p}_{\mathrm{x}}} 
\newcommand{\dotcoordy}{\dot{p}_{\mathrm{y}}} 
\newcommand{\eqdefn}{\triangleq} 
\newcommand{\sumdim}[2]{\sum_{i = 0}^{#2} {#1}^i}
\newcommand{\vectordim}[1]{\in \real^{#1}}
\newcommand{\matrixdim}[2]{\in \real^{{#1} \times {#2}}}
\newcommand{\E}{\mathbb{E}}
\newcommand{\OOO}{\mathcal{O}}
\newcommand{\redpower}[2]{#1^{\langle#2\rangle}}

\newcommand{\set}[1]{\left\{{#1}\right\}}
\newcommand{\setcomp}[2]{\left\{{#1}\,\middle\vert\,{#2}\right\}}
\newcommand{\trunc}[1]{\widetilde{#1}}
\newcommand{\truncdim}[2]{\widetilde{#1}^{(#2)}}
\newcommand{\err}[1]{e^{(#1)}}
\newcommand{\Err}[1]{\epsilon^{(#1)}}
\newcommand{\GErr}{\varepsilon_{P,t}}
\newcommand{\abs}[1]{| {#1} |}
\newcommand{\card}[1]{| {#1} |}
\newcommand{\norm}[1]{\| {#1} \|}
\newcommand{\captionspace}{\vskip -10pt}

\newcommand{\ini}{\textnormal{ini}}

\newcommand{\F}[1]{\mathfrak {F}({#1})}
\newcommand{\y}[2]{y_{#1}({#2})}
\newcommand{\A}[2]{\mathfrak {A}_{#1}({#2})}
\newcommand{\EE}[1]{\mathfrak{E}_{#1}}

\renewcommand{\phi}{\varphi}

\newcommand{\mytablecaption}[2]{\refstepcounter{table}\label{#2}
  \raisebox{7pt}{\normalsize{Table \thetable.\ #1}}}

\newcommand{\mycaption}[2]{\refstepcounter{figure}\label{#2}\raisebox{-7pt}
  {Fig. \thefigure.\hspace{3pt} #1}}

\newtheorem{ex}{Example}
\newtheorem{remark}{Remark}

\usepackage{subfig}

\usepackage[dcucite]{harvard}
\usepackage{balance}
\usepackage{times} 
\newcommand{\infnorm}[1]{\left\|{#1}\right\|_\infty} 

\newcommand{\rv}[1]{#1}  

\begin{document}

\begin{frontmatter}

\title{Moment Propagation of Polynomial Systems Through Carleman Linearization for Probabilistic Safety Analysis}


\thanks[footnoteinfo]{The authors are supported by ERATO HASUO Metamathematics for Systems Design Project No.~JPMJER1603, JST; 
S.\ Pruekprasert is also supported by Grant-in-aid No. 21K14191, JSPS. 
A.\ Cetinkaya is also supported by JSPS KAKENHI Grant Numbers~JP20K14771 and~JP23K03913. 
Part of the material in this paper was presented at the 21st IFAC World Congress, Berlin, Germany.}

\author[aist]{Sasinee Pruekprasert} \ead{s.pruekprasert@aist.go.jp},
\author[aist]{Jeremy Dubut} \ead{jeremy.dubut@aist.go.jp},
\author[chd]{Toru Takisaka} \ead{takisaka@uestc.edu.cn},
\author[nii,jfli]{Clovis Eberhart} \ead{eberhart@nii.ac.jp},
\author[sbhr]{Ahmet Cetinkaya} \ead{ahmet@shibaura-it.ac.jp}
\address[aist]{National Institute of Advanced Industrial Science and Technology, Tokyo, Japan}
\address[chd]{University of Electronic Science and Technology of China, Chengdu, China}
\address[nii]{National Institute of Informatics, Tokyo, 101-8430, Japan}
\address[sbhr]{Shibaura Institute of Technology, Tokyo, 135-8548, Japan}
\address[jfli]{Japanese-French Laboratory of Informatics, IRL 3527, CNRS, Tokyo, Japan}


\begin{abstract}
We develop a method to approximate the moments of a discrete-time
stochastic polynomial system. Our method is built upon Carleman linearization with
truncation. Specifically, we take a stochastic polynomial system with
finitely many states and   
transform it into
an infinite-dimensional system
with
linear deterministic dynamics, which describe the exact evolution of the moments of the
original polynomial system. 
We then truncate this deterministic system to obtain a finite-dimensional linear system,
and use it for moment approximation by iteratively propagating the moments
 along the finite-dimensional linear dynamics across time. 
We provide efficient online computation methods for this propagation scheme with several error bounds for the approximation.
Our results also show that precise values of certain moments \rv{at a given time step} can be obtained when the truncated system is sufficiently large.
Furthermore, we investigate techniques to reduce the offline computation load using reduced Kronecker power.  
Based on the obtained approximate moments and their errors, we also
provide hyperellipsoidal regions that are safe for some given probability bound. 
Those bounds allow us to conduct probabilistic safety analysis online
through convex optimization. We demonstrate our results on a logistic
map with stochastic dynamics and a vehicle dynamics subject to stochastic
disturbance.
\end{abstract}

\end{frontmatter}

\section{Introduction}

	\begin{figure*}
  {\small
	\begin{center}
    \newcommand{\hlength}{4.5}
    \newcommand{\vlength}{1}
   \resizebox{.9\textwidth}{!}
    { 
    \begin{tikzpicture}
		\node[draw,align=center] (1) at (0,0) {finite-dim.\\polynomial\\stochastic};
		\node[draw,align=center] (2) at (0.8*\hlength,0) {infinite-dim.\\linear\\stochastic};
		\node[draw,align=center] (3) at (1.6*\hlength,0) {infinite-dim.\\linear\\deterministic};
		\node[draw,align=center] (4u) at (3*\hlength,0.8\vlength) {finite-dim. linear deterministic};
		\node[draw,align=center] (4d) at 
        (3*\hlength,-0.7\vlength) {approximation error bounds};
		\node[draw,align=center] (5) at (4.3*\hlength,0) {prob.\\ellipsoid\\bounds};
      \node[align=center] () at ($($(5.west)!(4d.east)!(5.east)$)!0.25!(5.west)$) {Tail Prob. Approx.\\Section~\ref{sec:ellipsoid}};
		\path[->] (1) edge node[above]{Carleman lin.} node[below]{Section~\ref{sec:CarlemanLin}} (2)
		          (2) edge node[above]{Expectation} node[below]{Section~\ref{sec:MomentEq}} (3)
		          (3) edge node[above,sloped,align=center]{Truncation\\Section~\ref{sec:TruncatedSys}} (4u)
              (3) edge node[above,sloped,align=center]{Truncation Error}
                       node[below,sloped]{Section~\ref{sec:TruncationError} and ~\ref{subsec:approx_indiv}} (4d)
		          (4u) edge (5)
		          (4d) edge (5);
		\end{tikzpicture}}
		\mycaption{The different steps of the proposed method}{fig:steps}
	\end{center}
  }
	\end{figure*}

Uncertainty is one of the critical issues that make safety assurance of
cyber-physical systems a difficult task. Handling uncertainty in automated
driving systems is especially challenging, as motion planning algorithms are
required to take account of uncertainty in both the measurement and the
actuation mechanisms~\cite{BryR11rapidly}. Depending on the dynamical variations
in the environment and the vehicle-controller interaction, the actual trajectory
taken by a vehicle may differ substantially from a trajectory that is
precomputed based on nominal dynamics of the vehicle. Most motion planning
algorithms address this issue by assessing the collision risk for a set of
possible future trajectories of the vehicle and not just for the nominal
trajectory. Moreover, there are different approaches regarding how the
uncertainty is modeled. A well known approach is based on considering known
deterministic bounds on the disturbance and the measurement noise (see
\citeasnoun{althoff2014online}, and the references therein). Another approach is
to use stochastic system models to describe the effects of uncertainty.

The effect of uncertainty on the state of a stochastic system can be quantified
through the so-called \emph{mean and covariance propagation} method. This method
is especially effective for systems with linear dynamics where the uncertainty
is modeled as additive Gaussian noise. For those systems, if the initial state
distribution is Gaussian, the distribution of the state at any future time
remains Gaussian. Moreover, the mean and the covariance of the state can be
computed iteratively using linear equations. Previously,
\citeasnoun{BryR11rapidly} and \citeasnoun{BanzhafDSZ18footprints} used this
 together with Kalman filters in motion planning tasks.

While mean and covariance propagation is an effective method for linear Gaussian
systems, many cyber-physical systems involve nonlinearities and the 
noise cannot always be modeled to be an additive noise or a Gaussian one. 
 Our goal 
 is to develop
a moment propagation method for nonlinear systems where the noise can be non-additive
and non-Gaussian. Specifically, our propagation method can approximate the mean,
covariance, and higher moments of such systems. We consider discrete-time
stochastic nonlinear systems with polynomial dynamics, where the coefficients
are randomly varying. Our model can describe both additive and multiplicative
noise. Our propagation method is based on Carleman linearization with
truncation~\cite{steeb1980non}. By following the Carleman linearization
technique, we obtain a new dynamical system that describes how the Kronecker
powers of the system state evolve. This system is linear, but it is
infinite-dimensional. The expectation of the states of this linear system
correspond to all moments of the original nonlinear system. By truncating the
expected dynamics of the linear system at a certain truncation limit, we obtain
a finite-dimensional linear system. Using the equations of this
finite-dimensional system, our method iteratively computes the approximate moments of the nonlinear system up to the provided truncation limit.

The mean, covariance, and higher moments of a system can provide useful
statistical information regarding possible future state trajectories and can be
used in safety verification. In this paper, we use our moment approximation
method as a basis to develop a stochastic safety verification framework. Our
framework uses tail probability bounds~\cite{gray1991general}.   
 We provide hyperellipsoidal safety region, for which the probability of the system state to be outside that region is within a given probability bound.
Such safety regions are characterized using positive-semidefinite matrices. 
Previously, a scenario-based approach was used by
\citeasnoun{shen2020probabilistic} to obtain similar safety regions for
vehicles. This scenario-based approach is different from our moment-based approach:
it relies on generating sufficiently many sample state trajectories using the
information of the distributions of random elements in the dynamics. In our
approach, we do not need to generate sample trajectories.

In this paper, we propose a novel tail probability-based safety verification framework consisting of two phases.
The first one is an offline computation phase 
to compute the approximate moment dynamics and compute the upper bounds for their errors. 
The second phase is a fast online computation to propagate the moments and their errors through the truncated dynamics, then use them to compute probabilistic ellipsoid bounds.
A key ingredient in our framework is to efficiently compute bounds on moment approximation errors. 
First, we present the exact approximation error for a particular moment as a
sum.
However, for the exact computation, the number of terms is a quickly growing function of the dimension of the system and the time. 
We provide efficient online computation methods for this propagation scheme with several error bounds for the approximation.
Furthermore, we 
develop novel techniques to reduce the offline computation load using reduced Kronecker power.


Remark that, similarly to~\citeasnoun{BryR11rapidly}, 
this paper focuses on discrete-time systems. 
 We could also 
take into account the error due to going from continuous-time to discrete-time dynamics by 
incorporating Lagrange remainders in our analysis, but we decided against it to keep the 
presentation of our method simple.
Remark also that extending the analysis of~\citeasnoun{forets2017explicit} on continuous-time systems 
by adding general 
noise with both additive and multiplicative noise terms is not easily achievable, which motivated the present work.

We present two case studies to demonstrate various aspects of our method. Our
first case study is on a scalar stochastic logistic map.
In particular, we
consider uniformly distributed growth/decay rates in the dynamics. We run moment
approximation and tail probability-based safety analysis. In the safety
analysis, we check the conservativeness of our tail-probability bounds through
Monte Carlo simulations. We observe that our obtained safety regions are quite
tight for small durations. We also observe that our method is advantageous in
terms of speed. It does not require drawing random variables, and as a result, it
is much faster than Monte Carlo simulations even with small number of samples.
Another advantage of our method is that it provides a theoretical guarantee for
 obtaining the safety regions.

Our second case study is on a vehicle with kinematic bicycle model. First, we
use second-order Taylor expansion to obtain a discrete-time nonlinear polynomial
system describing the vehicle. Then we use our approximate moment propagation
approach to study how the vehicle's future states are affected by the
uncertainty in the acceleration and in the initial vehicle state.
Our approach uses the moments of the initial state of the vehicle and the
moments for the acceleration to approximate future moments of the state. Here
the uncertainty in the initial state is caused by the measurement noise and thus
initial state moments are derived from the statistical properties of the
measurement error. On the other hand uncertainty in the acceleration is related
to actuators and the environment.
In a typical control loop, our approach can be applied in a
receding horizon fashion.
In particular, computation is done by
first obtaining new state measurements which are used for specifying the moments
of the new ``initial'' state; then these moments are propagated to
approximately compute the moments of the state associated with future time
instants.

The rest of the paper is organized as follows (see also 
Figure~\ref{fig:steps}). 
Section~\ref{sec:relate} presents related work on Carleman linearization. 
In Section~\ref{sec:Carleman}, we start from a finite-dimensional
stochastic polynomial system, on which we apply the Carleman linearization
technique to get an infinite-dimensional linear stochastic system.
Then, we
 turn the linear stochastic system into a deterministic one by taking expectation. 
We present our
truncation-based moment approximation method and study the error it introduces
in Section~\ref{sec:Approximation}. 
Then, 
in Section~\ref{sec:ellipsoid}, we provide probabilistic
ellipsoid bounds based on approximations
of the first and second moments.
Section~\ref{section: reduced Kronecker} presents techniques to reduce the offline computation load using the concept of reduced Kronecker power.
We provide two numerical examples to demonstrate the applicability of our approach in
Section~\ref{sec:numerical-example}. Finally, in Section~\ref{sec:Conclusion},
we conclude the paper.

\section{Related Work}\label{sec:relate}
The characterization of the dynamics under truncated Carleman linearization and
the truncation error analysis in Section~\ref{sec:Carleman} appeared previously
in our preliminary conference report \cite{pruekprasert2020moment}. 

In this
paper, we present new results on tail-probability analysis 
and a two-step approximation procedure for  confidence ellipsoid bounds of the system state
in Section~\ref{sec:ellipsoid}. We
investigate techniques to improve the computational efficiency of our methods in
Section~\ref{section: reduced Kronecker}.  
We also provide a new upper bound for the approximation error from truncation in Section~\ref{subsubsec:boundK}, using partial exact computation with indices on moment coordinates. 
We illustrate these new results and techniques 
in
Section~\ref{sec:numerical-example} through new simulations for a case study on
vehicle dynamics.

The Carleman linearization technique is well-known in the nonlinear systems
literature. For deterministic systems, it has been used for approximation of
nonlinear models by linear
ones~\cite{bellman1963some,steeb1980non,tuwaim1998discretization}. Recently,
\citeasnoun{amini2019carleman} used a Carleman linearization approach to
design state feedback controllers for continuous-time nonlinear polynomial
systems, 
 \citeasnoun{hashemian2019feedback} used Carleman linearization in
model predictive control of continuous-time deterministic nonlinear systems, \rv{and
\citeasnoun{amini2021error} obtained error bounds for Carleman linearization with truncation}. Furthermore, \citeasnoun{amini2020quadratization} proposed 
control
frameworks based on Carleman linearization.  
There, an approach using Carleman linearization with
truncation allows the authors to approximately represent the
Hamilton-Jacobi-Bellman equation (arising in optimal control of nonlinear
systems) as an operator equation in quadratic form. This representation is then
used to derive the approximate value function of the optimal control system as a
quadratic Lyapunov function. The main advantage of the method is that it yields
an iterative procedure to approximate the value function that converges to the
true value. In our work, we do not investigate the optimal control problem and
we use Carleman linearization with truncation to approximate moments of
stochastic systems instead of value functions. 

There are fewer results on Carleman linearization for stochastic systems.
Specifically, \citeasnoun{wong1983carleman}
investigated bilinear noise terms in Ito-type stochastic differential equations
and used Carleman linearization to describe how the moments of the stochastic
differential equations evolve. Furthermore, \citeasnoun{rauh2009carleman}
considered continuous-time nonlinear systems with additive noise and used
a Carleman linearization technique in conjunction with a series expansion approach
to approximate such systems with discrete-time linear systems. They then
demonstrated that a Kalman filter can be used on the obtained linear systems for
state estimation. In addition, \citeasnoun{cacace2014carleman} and
\citeasnoun{cacace2019optimal} proposed Carleman linearization-based
sampled-data filters for nonlinear stochastic differential equations driven by
Wiener noise.
More recently, \citeasnoun{jasour2021moment} used a technique akin to
Carleman linearization in
moment propagation of a class of discrete-time mixed trigonometric polynomial
systems with probabilistic disturbance. They show that several systems of
interest are mixed trigonometric polynomial of degree $1$, for which moment
propagation without error is computationally feasible.
The truncation error analysis in our paper is partly motivated by
the work of \citeasnoun{forets2017explicit}, which derives tight approximation error
bounds for the truncated Carleman linearization of deterministic continuous-time
systems. We note that deriving linear representations of finite-dimensional
nonlinear systems can also be achieved through the use of Koopman
operators~\cite{goswami2017global,mesbahi2019modal}. A Koopman operator approach
has recently 
been used in model predictive control of vehicles with
deterministic nonlinear dynamics \cite{cibulka2020model}.

\textbf{Notations.} 
We use $\real$ and $\nonnegativeInt$ for the sets of real numbers and non-negative integers, respectively.
We write $1$ for the $(1\times1)$-matrix $\begin{bmatrix}
1 \end{bmatrix}$.
We denote by $M_1\otimes M_2$ the Kronecker product of two matrices $M_1$ and $M_2$.
 We use $M^{[k]}$ to denote the $k$th Kronecker power of the matrix $M$, given by $M^{[0]} = 1 $ 
and
	$M^{[k]}=M^{[k-1]}\otimes M$ for $k > 0$.
For a matrix
$M$ 
 of random variables, we write 
$\E[M] $ for its expectation.

%
%

\section{Carleman Linearization for Stochastic Polynomial Systems}
\label{sec:Carleman}

In this section, we present the Carleman linearization technique to transform
  a finite-dimensional discrete-time stochastic polynomial system into
  an infinite-dimensional one, then take expectation to get a
  deterministic system that describes the evolution of all moments of
  the system state.

\subsection{Discrete-Time Stochastic Polynomial Systems}
	
	We consider the following finite-dimensional discrete-time stochastic polynomial system
	\begin{align}
	\label{eq:system}
	\begin{split}
	x(t+1) &= \sum\limits_{i = 0}^{\dimsys} F_i(t)x^{[i]}(t), \quad t\in\nonnegativeInt,	
	\\
	x(0) &=x_\ini,
	\end{split}
	\end{align}
	where $x(t)\in \real^n$ is the state vector and $F_i(t) \in
  \real^{n\times n^i}$, $i\in\coeffindices$, are 
  \rv{matrix-valued stochastic processes.}
  More precisely, we assume given a probability space $\Omega$,
  then 
  $x_\ini$ (\emph{resp.}\ 
  $x(t)$, $F_i(t)$) is a measurable
  function from $\Omega$ to $\real^n$ (\emph{resp.}\ $\real^n$,
  $\real^{n\times n^i}$). 
  We consider the scenarios where all $F_i(t)$'s have known distributions.
	
	Systems of the form \eqref{eq:system} can be used to model dynamics with both 
	additive and multiplicative noise terms. The vector $F_0(t)\in \real^n$ 
	represents additive noise, as $x^{[0]}(t)=1$, while the matrices
	$F_1(t),\ldots,F_{\dimsys}(t)$ characterize the effects of multiplicative noise. Then, we define
	\begin{equation} \label{eq:F}
	\F{t} \eqdefn
	\begin{bmatrix}
	F_0(t) & F_1(t) &\cdots & F_{\dimsys}(t)
	\end{bmatrix}
	\matrixdim{n}{\sumdim{n}{\dimsys}},
	\end{equation}
which is a matrix with $n$ rows and $\sum_{i=0}^{\dimsys} n^i$ columns.
 
 We make the following assumptions concerning the coefficient matrices and the random initial state $x_\ini$.

	\begin{assum} \label{assum:f-independence}
    The vector $x_\ini$ and the matrices ${\F{t}}$, $t\in\nonnegativeInt$, are all independent.
	\end{assum}
	
	\begin{assum} \label{assum:f-identical}
    The matrices $\setcomp{\F{t}}{t\in\nonnegativeInt}$ are identically
    distributed.
  \end{assum}
  
  
	Note that Assumptions~\ref{assum:f-independence} and~\ref{assum:f-identical}
	 are not overly restrictive and they hold in fairly general situations as we 
	 discuss in Section~\ref{sec:numerical-example}. 
	 Notice also that under Assumption~\ref{assum:f-independence}, matrices 
	 $F_i(t)$ and $F_j(t)$ are still allowed to statistically depend on each other. 
	 Moreover, Assumption~\ref{assum:f-identical} allows us to obtain a 
	 ``time-invariant'' method to compute moments, further yielding computational 
	 advantage.
	 

\begin{ex}\label{ex:1}
  Consider 
\begin{align*} 
x_1(t+1) &= a(t) x_1(t) x_2(t), &x_1(0) = x_{\ini, 1},\\
x_2(t+1) &= a(t) (x_1(t) + x_2(t)), &x_2(0) = x_{\ini, 2},
\end{align*}
where $x_{\ini, 1}$, $x_{\ini, 2}$, and $a(t)$, $t \in \mathbb{N}$, are independent and identically distributed 
 random variables. 
By using $x(t) = \begin{bmatrix} x_1(t) & x_2(t) \end{bmatrix}^\intercal$, $ F_0(t) =\begin{bmatrix} 0& 0\end{bmatrix}^\intercal$,  $ F_1(t) = \begin{bmatrix} 0& 0\\a(t)& a(t) \end{bmatrix} $, and $F_2(t) = \begin{bmatrix} 0 & a(t) & 0 & 0 \\ 0& 0 & 0 & 0 \end{bmatrix}$,
we can rewrite the system as in the form of \eqref{eq:system}, i.e.,
\begin{align}\label{ex:1 eq1}
x(\rv{t+1})&=F_0 x^{[0]}(t) + F_1 x^{[1]}(t) + 
F_2 x^{[2]}(t).
\end{align}
\end{ex}	

Our objective is to approximate the moments of the state of 
 \eqref{eq:system}, 
which we will use in Section~\ref{sec:ellipsoid} to compute
 probabilistic safety areas for the state $x(t)$ at a given time $t$.
 For Example~\ref{ex:1}, given $j$ and $t$, our objective is to approximate the moment $\mathbb{E}[x^{[j]}(t)]$. 
Note that the initial moments $\mathbb{E}[x^{[j]}(0)]$ can be computed from the probability distribution of $x_\ini$.
	
	\subsection{Carleman Linearization}
	\label{sec:CarlemanLin}
	We use Carleman linearization to obtain an infinite-dimensional
  linear system that describes the evolution of the Kronecker powers
  of the state vector $x(t)$.
	By defining
	\begin{align}
	{\y{k}{t}} \eqdefn
	\begin{bmatrix}
	x^{[0]}(t)^\intercal & x^{[1]}(t)^\intercal  &
	\cdots & x^{[k]}(t)^\intercal
	\end{bmatrix}^\intercal, 
	 \label{eq:ydef}
	\end{align}
	we can rewrite the dynamical system \eqref{eq:system}  using \eqref{eq:F} and \eqref{eq:ydef} as
	\begin{align*}
	x(t+1) &= 
	\begin{bmatrix}
	F_0(t) &\cdots & F_{\dimsys}(t)
	\end{bmatrix}
	\begin{bmatrix}
	x^{[0]}(t)^\intercal &&
	\cdots && x^{[\dimsys]}(t)^\intercal
	\end{bmatrix}^\intercal
	\\
	&=
	\F{t}\, \y{\dimsys}{t}. 
	\end{align*} 
	Therefore, for all $j \in \nonnegativeInt$, we have
	\begin{align}
	  x^{[j]}(t+1) = (\F{t} \,\y{\dimsys}{t})^{[j]}.
	  \label{eq Fy}
	\end{align} 
	Then, we introduce Theorem~\ref{theorem H} to compute $(\F{t} \,\y{\dimsys}{t})^{[j]}$.	
\begin{thm}
\label{theorem H}
 Consider $\F{t}$ and $\y{\dimsys}{t}$ given in equations~\eqref{eq:F} and~\eqref{eq:ydef}, respectively. We have
\begin{align*}
	  \big(\F{t}&\,\y{\dimsys}{t}\big)^{[j]} \\
       &= \sum_{k=0}^{j\dimsys} \left(
      \smashoperator[r]{\sum_{(i_1, \ldots,i_l) \in H_{j,k}}}
      F_{i_1}(t) \otimes \cdots \otimes F_{i_j}(t) \right) x^{[k]}(t),
	\end{align*}
	for $j\in\nonnegativeInt$, where 
	\begin{equation*} 
H_{j,k} \eqdefn \setcomp{(i_1, \ldots,i_j) }{\sum\limits_{l=1}^j i_l
  = k \text{ and } 0 \leq i_l\leq \dimsys}.
\end{equation*}	
\end{thm}	
The proof of Theorem~\ref{theorem H} is presented in Appendix~\ref{appendixA}.

By \eqref{eq Fy} and Theorem~\ref{theorem H}, 
we obtain
 the infinite-dimensional linear system
	\begin{align}
    x^{[j]}(t+1) &= \sum_{k=0}^{j\dimsys} \left(
      \smashoperator[r]{\sum_{(i_1,\ldots,i_j) \in H_{j,k}}}
      F_{i_1}(t) \otimes \cdots \otimes F_{i_j}(t) \right) x^{[k]}(t),\nonumber
      \\
	x^{[j]}(0) &= x^{[j]}_\ini,\label{eq:infsys}
	\end{align}
	which describes the evolution of all Kronecker powers $x^{[0]}(t),
	x^{[1]}(t), \ldots$ of the state $x(t)$.
	
	In order to give a simpler description of the system, we introduce the
	matrices $A_{j,k}(t) \in \real^{n^j \times n^k}$ given by
	\begin{align}\label{eq:defA}
	 A_{j,k}(t) \eqdefn \sum_{(i_1,\ldots,i_j)\in H_{j,k}} F_{i_{1}}(t)
	\otimes \cdots \otimes F_{i_j}(t).
	\end{align}
	Note in particular that 
 $H_{0,0}$ only contains the empty tuple and  $A_{0,0}(t) = 1$, as it is the identity of $\otimes$. We also have 
	\begin{align} \label{A0}
A_{j,k}(t) = 0 \text{ if }  k > j \dimsys,
	\end{align}
since $H_{j,k}$ is then empty.
	We then introduce, for all $N, M \in \mathbb{N}$, the matrix {$\A{{N,M}}{t}$} defined by blocks as:
	\begin{align} \label{eq:A}
	\A{{N,M}}{t} & \triangleq\left[\begin{array}{ccc}
	A_{0,0}(t) & \cdots & A_{0,M}(t)\\
	\vdots & \ddots & \vdots\\
	A_{N,0}(t) & \cdots & A_{N,M}(t)
	\end{array}\right], 
	\end{align}
which is a matrix of $\sum_{i=0}^{N} n^i$ rows and $\sum_{i=0}^{M} n^i$ columns.
 From \eqref{eq:infsys}, for any $k \in \mathbb{N}$, we have
	\begin{align*} 
	\begin{split}
	\begin{bmatrix}
	x^{[0]}(t+1)\\ 
	\vdots 
	\\
	x^{[k]}(t+1) 
    \end{bmatrix} 
    =
    \left[\begin{array}{ccc}
	A_{0,0}(t) & \cdots & A_{0,k \dimsys}(t)\\
	\vdots & \ddots & \vdots\\
	A_{k,0}(t) & \cdots & A_{k,k \dimsys}(t)
	\end{array}\right] 
	\begin{bmatrix}
	x^{[0]}(t) \\ 
	\vdots 
	\\
	x^{[k]}(t) \\  	
	\vdots 
	\\
	x^{[k \dimsys]}(t)
    \end{bmatrix}, 
    \end{split}
	\end{align*}
	which can be rewritten using \eqref{eq:ydef} and \eqref{eq:A} as
	\begin{align} \label{eq A k x k^2}
	\y{k}{\rv{t+1}} 
	=
	\A{k, k\dimsys}{t} \, \y{k \dimsys}{t}. 
	\end{align}
Notice that 
 is not closed: we need $\y{k \dimsys}{t}$ 
  in order to compute 
	$\y{k }{t+1}$. 
	 As a result, the linear model describing the system cannot be written using a finite state space. 
  \rv{However, in Section~\ref{sec:Approximation},  we can restrict
   this system  to a finite dimensional one as we are interested in a finite number of moments on a finite time horizon.}
	
	\begin{ex}\label{ex:2}
	For the system~\eqref{ex:1 eq1} in Example~\ref{ex:1},  
	\begin{align*} 
	x^{[2]}&(t+1) = \\&\big(F_0(t) x^{[0]}(t) + F_1(t) x^{[1]}(t) + F_2(t) x^{[2]}(t) \big)^{[2]}.
	\end{align*}  
By~\eqref{eq A k x k^2}, we have
\begin{equation}\label{eq Ex2}
\y{2}{t+1} = \A{2, 4}{t}\y{4}{t},
\end{equation} 
where\\
$  \A{2, 4}{t} = \begin{bmatrix}
	1 & 0 & 0 & 0  & 0\\
	F_0 & F_1 & F_2 & 0 & 0\\
	F_0^{[2]}
		& \begin{small}
		\begin{smallmatrix}F_0 \otimes F_1\\ + F_1 \otimes F_0\end{smallmatrix} 
		\end{small}
		&
        \begin{small}
		\begin{smallmatrix}{F_0 \otimes F_2 + F_1^{[2]}}\\ {+ F_2 \otimes F_0}\end{smallmatrix} 
		\end{small}
        & 
		\begin{small}
		\begin{smallmatrix}F_1 \otimes F_2\\ + F_2 \otimes F_1\end{smallmatrix} 
		\end{small} 
		& F_2^{[2]}
	\end{bmatrix}  $.

Note that we omitted time dependence of $F_0$, $F_1$ and $F_2$ for brevity. 
Since $ F_0(t) =\begin{bmatrix} 0& 0\end{bmatrix}^\intercal$, \eqref{eq Ex2} can also be written as  
	\begin{alignat*} {2}
	&\begin{bmatrix} 1 \\ x^{[1]}
	(t+1) \\ x^{[2]}(t+1)\end{bmatrix} =&&
	\begin{bmatrix}
	1 & 0 & 0 & 0  & 0\\
	0 & F_1 & F_2 & 0 & 0\\
	0 & 0
		& F_1^{[2]}& 
		\begin{small}
		\begin{smallmatrix}F_1 \otimes F_2\\ + F_2 \otimes F_1\end{smallmatrix} 
		\end{small} 
		& F_2^{[2]}
	\end{bmatrix}   
	\begin{bmatrix}
	1 \\ x^{[1]}(t) \\ x^{[2]}(t) \\ x^{[3]}(t) \\ x^{[4]}(t)	
	\end{bmatrix}.
	\end{alignat*} 
	\end{ex}
	Recall that our objective is to approximate  $\mathbb{E}[x^{[j]}(t)]$	for given $j$ and $t$.
	In the next section, we will consider a vector $\mathbb{E}[\y{k}{t}]$ which includes all moments $\mathbb{E}[x^{[j]}(t)]$ where $j \leq k$.

	\subsection{Moment Equations}
  \label{sec:MomentEq}

	We now derive the deterministic system that describes the evolution of
	the moments of $x(t)$ by taking expectation in~\eqref{eq A k x k^2}.
	This gives
	\begin{align*}
	\mathbb{E}[\y{k}{t+1}] &= \mathbb{E}[
	\A{k, k \dimsys}{t} \y{k\dimsys}{t}].
	\end{align*}

	By iteration of \eqref{eq A k x k^2}, 
	we get that $\y{k\dimsys}{t}$, $t\in \nonnegativeInt$, is given by
	\begin{align}\label{eq:ykdsatt}
	\y{k\dimsys}{t}& = 
	\A{k\dimsys,k\dimsys^2}{t-1} \cdots
	\A{k\dimsys^{t},k\dimsys^{t+1}}{0}  \cdot \y{{k\dimsys^{t+1}}}{0}. 
	\end{align}
	It follows from Assumption~\ref{assum:f-independence} that $\A{k,k\dimsys}{t}$
	 and $\y{k\dimsys}{t}$  in \eqref{eq A k x k^2} are mutually
	 independent. To see this, observe that $\A{k,k\dimsys}{t}$ is composed of 
	 the matrices $F_i(t)$, which are independent of $x_\ini$ and 
	 $\F{t-1},\ldots,\F{0}$, which determine $\y{k\dimsys}{t}$
	 as given by \eqref{eq:ykdsatt}. It then follows that
	\begin{align*}
\mathbb{E}[\y{k}{t+1}] &= \mathbb{E}[\A{k,k\dimsys}{t}]\mathbb{E}[\y{k\dimsys}{t}].
	\end{align*}
Here again, to give a simpler description of the system, we introduce
	new matrices.
	Notice that the coefficients of the matrix $\mathbb{E}[\A{N,M}{t}]$ are
	products of the moments of coefficients of $\F{t}$ and thus independent
	of $t$ by Assumption~\ref{assum:f-identical}.

	Thereby, we use the notations: 
	\begin{align}
	\begin{split}	
	\label{eq:defE}
	E_{i,j} &= \mathbb{E}[A_{i,j}(t)] \matrixdim{n^i}{n^j}, \\
	\EE{N,M} &= \mathbb{E}[\A{N,M}{t}]\matrixdim{\sumdim{n}{N}}{\sumdim{n}{M}},
	\end{split}
	\end{align}
	emphasizing the fact that
	both are independent of the time.
	
	Then, from \rv{\eqref{eq:infsys}}, we can obtain the following system.
	\begin{align} \label{eq:momentsequation}
	\begin{split}
	\mathbb{E}[x^{[j]}(t+1)] &= \sum\limits_{k = 0}^{j\dimsys} E_{j,k}\mathbb{E}[x^{[k]}(t)],\quad t\in \nonnegativeInt, \\
	\mathbb{E}[x^{[j]}(0)] &= \mathbb{E}[x_\ini^{[j]}].
	\end{split}
	\end{align}
	
	From \eqref{eq:defA}, the matrices $E_{j,k}$ only depend on the moments $\E[F_i(\cdot)]$, which are independent of $t$ as discussed above.
	As a result, the matrices $E_{j,k}$ can be computed offline.

\section{Moment Approximation through Truncation}
	\label{sec:Approximation}
	
	We will now introduce an approach that allows us to compute
	\emph{approximations} of the moments of $x(t)$.
	This truncation approach is critical, as an exact computation of the
	moments is impossible from a practical point of view.
	Indeed, by~\eqref{eq:ydef}, computing the first $k$ moments
	at time $t$ amounts to computing $\mathbb{E}[\y{k}{t}]$. So, by
  	iteration of~\eqref{eq:momentsequation},
  \begin{equation}
    \mathbb{E}[\y{k}{\rv{t+1}}] = \EE{k,k\dimsys}\cdots \EE{k\dimsys^{t-1},k\dimsys^t}\E[\y{k\dimsys^t}{0}]\rlap{.}
    \label{eq:momentsequationiterated}
  \end{equation}
	As
	this indicates, exact computation of the first $k$ moments requires
	the knowledge of matrices $\EE{k,k\dimsys}$, $\EE{k\dimsys,k\dimsys^2}$,
	$\ldots$, $\EE{k\dimsys^{t-1},k\dimsys^t}$ of exponentially increasing
	dimensions, making any practical computation unrealistic.

Thus, in the next section, we compute approximations of the moments of $x(t)$ by trucating the system~\eqref{eq:momentsequation}.

	\subsection{Approximate Moments and the Truncated System}
	\label{sec:TruncatedSys}

	In this section, we define the system that we use to compute
	approximations of the moments of $x(t)$.
	
	We fix the truncation limit $\dimtru \in \mathbb{N}$, \rv{and define approximate moments 
 $\truncdim{x}{i}(t) \in \real^{n^i}$},
	$i\in{1,\ldots,\dimtru}$, by
	\begin{align}\label{eq tilde x}
    & \begin{bmatrix}
      1 & {\truncdim{x}{1}(t)}^\intercal &  
      \cdots  &  {\truncdim{x}{\dimtru}(t) }^\intercal
    \end{bmatrix}^\intercal \nonumber \\
    & \quad =
    {\EE{\dimtru,\dimtru}^t}
    \begin{bmatrix}
      1 &  {\mathbb{E}[{x_\ini}]}^\intercal & 
      \cdots  & {\mathbb{E}[x_\ini^{[\dimtru]}]}^\intercal
    \end{bmatrix}^\intercal.
	\end{align}
 
  Notice that \eqref{eq tilde x} follows the same pattern as~\eqref{eq:momentsequationiterated}, but
  use the square matrix $\EE{\dimtru,\dimtru}$.
	Here, the vector $\truncdim{x}{i}(t)$ represents an approximation of the
	moment $\mathbb{E}[x^{[i]}(t)]$ that is computed using only our
	knowledge of the first $\dimtru$ moments of $x_\ini$.
 \rv{Note that the superscript $(i)$ is only for notation purpose, and it has no relation to the Kronecker power.}
	
	By letting
	\begin{align}\label{eq tilde y}
	\trunc{y}(t)\eqdefn[\begin{array}{ccccc}
	1 & \truncdim{x}{1}(t)^{\intercal} & \truncdim{x}{2}(t)^{\intercal} & \cdots & \truncdim{x}{\dimtru}(t)^{\intercal}\end{array}]^{\intercal},
	\end{align} 
  we obtain what we call the ``truncated system'', which is a discrete-time 
	linear time-invariant system given by
	\begin{align}\label{eq:truncatedSystem}
	\begin{split}
	\trunc{y}(t+1) &= \EE{\dimtru,\dimtru}\trunc{y}(t), \quad t\in\nonnegativeInt, \\
	\trunc{y}(0) &= [\begin{array}{cccc}
	1 & \E[x_\ini(t)]^{\intercal} & \cdots & \E[x_\ini^{[\dimtru]}(t)]^{\intercal}
	\end{array}]^{\intercal}.
	\end{split}
	\end{align}
	The truncated system allows us to iteratively compute approximations
  of the moments of $x(t)$ at consecutive time instants.
	Moreover, the approach only requires an offline computation of the
	matrix $\EE{\dimtru, \dimtru}$.

	\begin{ex}\label{ex:3}
	By considering $\dimtru = 2$, we may approximate the first and and the second moments of the system in Example~\ref{ex:2} by the following truncated system.
	\begin{align*}
	\trunc{y}(t+1) &= [\begin{array}{ccccc}
	1 & \truncdim{x}{1}(t+1)^{\intercal} & \truncdim{x}{2}(t+1)^{\intercal})\end{array}]^{\intercal}\\
	&=
	\mathbb{E}\begin{bmatrix}\begin{bmatrix}
	1 & 0 & 0 \\
	0 & F_1(t) & F_2(t) \\
	0 & 0 & \begin{small} F_1(t) \otimes F_1(t)\end{small}  
	\end{bmatrix} \end{bmatrix} 
	\begin{bmatrix}
	1 \\ \truncdim{x}{1}(t)^{\intercal} \\ \truncdim{x}{2}(t)^{\intercal}
	\end{bmatrix}\\
	&=\EE{2,2}\,\trunc{y}(t),
	\end{align*}	
	where $\truncdim{x}{1}(0)^{\intercal} = \mathbb{E}[x_\ini(t)]^{\intercal}$
	and $\truncdim{x}{2}(0)^{\intercal} = \mathbb{E}[x_\ini^{[2]}(t)]^{\intercal}$.
	\end{ex}
	
	\subsection{Computation of Truncation Errors}
  \label{sec:TruncationError}
	
	We now consider the error due to the truncation.
	We consider $j_0 \in \{0,\ldots,\dimtru\}$, and
	let $\err{j_0}(t)\in\real^{n^{j_0}}$ denote the error of the $j_0$-th
	moment, that is,
	\begin{align}\label{eq:error}
	\err{j_0}(t) \eqdefn \mathbb{E}[x^{[{j_0}]}(t)] - \truncdim{x}{j_0}(t).
	\end{align}
	
	First, by \eqref{eq:momentsequation}, we have
	\begin{align*}
	\mathbb{E}[x^{[j_0]}(t)]
	&=\sum_{j_1=0}^{j_0\dimsys} E_{j_0,j_1}\mathbb{E}[x^{[j_1]}(t-1)]\nonumber\\
	&=\sum_{j_1=0}^{j_0\dimsys} E_{j_0,j_1}\sum_{j_2=0}^{j_1\dimsys} E_{j_1,j_2} \mathbb{E}[x^{[j_2]}(t-2)]\nonumber\\
	&=\sum_{j_1=0}^{j_0\dimsys} E_{j_0,j_1}\sum_{j_2=0}^{j_1\dimsys} E_{j_1,j_2} \cdots
	\sum_{j_t=0}^{j_{t-1}\dimsys} E_{j_{t-1},j_t}\mathbb{E}[x_\ini^{[j_t]}],
	\end{align*}
	
	and similarly, by \eqref{eq tilde x},
	\begin{align*}
	\truncdim{x}{j_0}(t)
	=\sum_{j_1=0}^{\dimtru} E_{j_0,j_1}\sum_{j_2=0}^{\dimtru} E_{j_1,j_2} \cdots
	\sum_{j_t=0}^{\dimtru} E_{j_{t-1},j_t}\mathbb{E}[x_\ini^{[j_t]}].
	\end{align*}
	
	From \eqref{A0} and \eqref{eq:defE}, we can observe that $E_{j,k} = 0$ for $k>j\dimsys$. Therefore, we obtain from the
	two equations above,
	\begin{align}
	&\err{j_0}(t) = \mathbb{E}[x^{[j_0]}(t)]  - \truncdim{x}{j_0}(t) \nonumber\\
	&\, = \smashoperator{\sum_{j_1=\dimtru+1}^{j_0\dimsys}} E_{j_0,j_1}
	\smashoperator{\sum_{j_2=0}^{j_1\dimsys}} E_{j_1,j_2} \cdots
	\smashoperator{\sum_{j_t=0}^{j_{t-1}\dimsys}} E_{j_{t-1},j_t}
	\mathbb{E}[x_\ini^{[j_t]}]\nonumber\\
	&~ + \smashoperator{\sum_{j_1=0}^{\dimtru}} E_{j_0,j_1}
	\smashoperator{\sum_{j_2=\dimtru+1}^{j_1\dimsys}} E_{j_1,j_2} \cdots
	\smashoperator{\sum_{j_t=0}^{j_{t-1}\dimsys}} E_{j_{t-1},j_t}
	\mathbb{E}[x_\ini^{[j_t]}]\nonumber\\
	&~ + \cdots\nonumber\\
	&~ + \smashoperator{\sum_{j_1=0}^{\dimtru}} E_{j_0,j_1} \cdots
	\smashoperator{\sum_{j_{t-1}=0}^{\dimtru}} E_{j_{t-2},j_{t-1}}
	\smashoperator{\sum_{j_t=\dimtru+1}^{j_{t-1}\dimsys}} E_{j_{t-1},j_t}
	\mathbb{E}[x_\ini^{[j_t]}].
	\label{eq:ej0}
	\end{align}
	
	As an immediate consequence, we get the following:
	\begin{prop} \label{prop:exact-computation}
		Consider the truncated approximation of the moments of
		system~\eqref{eq:system} with truncation limit $\dimtru\in
		\nonnegativeInt$.
    If $j_0 \dimsys^t \leq \dimtru$, then $\truncdim{x}{j_0}(t) = \E
    [x^{[j_0]}(t)]$.
	\end{prop}
	\begin{proof}
	We show that, if $j_0 \dimsys^t \leq \dimtru$, then
    $\err{j_0}(t) = 0$ and the proposition holds.
    To this end, it is enough to show that  all lines of~\eqref{eq:ej0} are equal to $0$ for any sequence $j_1,\ldots,j_t$ of relevant indices.
    Let us pick any $i$th line and any sequence $j_1,\ldots,j_t$.
     It is enough to show that there exists $k \in \{1, \ldots t\}$ where $j_{k} > j_{k-1} \dimsys$, so that
    $E_{j_{k-1}, j_{k}} = 0$ by \eqref{A0}, which also makes the $i$th line equal to $0$.  
For the sake of contradiction, we assume that
$j_{k} \leq j_{k-1} \dimsys$ for all $k \in \{1, \ldots t\}$, which makes 
$j_{i-1} \leq j_{i-2} \dimsys \leq j_{i-3} \dimsys^2 \leq \ldots \leq j_0 \dimsys^{i-1}\leq j_0\dimsys^t$.
However, notice that we have $j_{i-1} \geq \dimtru + 1 > \dimtru \geq j_0 \dimsys^t$, which is a contradiction that proves the desired result.
	\end{proof} 
	This shows that, for large values of truncation limit $\dimtru$, the
	proposed method computes \emph{exact} moments $\E[x^{[j_0]}(t)]$ for small enough
	$j_0$ and $t$.
	Note that this is due to the \emph{discrete-time} nature of the
	finite-dimensional polynomial system~\eqref{eq:system}.
	In the continuous-time case, approximation errors cannot be
	avoided in general~\cite{forets2017explicit}.
	
	Since $\E[x^{[j_0]}(t)] = \truncdim{x}{j_0}(t) + \err{j_0}(t)$,   
	if
	$\err{j_0}(t)$ could efficiently be computed, then so would
	$\E[x^{[i]}(t)]$, and there would be no need in using the truncated
	system.  
	However, the exact value of $\err{j_0}(t)$ is generally hard to compute.
	Therefore, in the following section, we provide upper bounds for $\err{j_0}(t)$.

  \subsection{Approximation of Error Bounds}
	\label{subsec:approx_indiv}
	
	We now investigate the approximation error introduced by truncation.
	For a given $j_0 \in \{0,\ldots,\dimtru\}$, 
	our goal is to obtain an upper bound on  $\infnorm{\err{j_0}(t)}$, 
  which is the error
  on the $j_0$-th moment introduced by the truncation (where $\infnorm{\cdot}$ denotes the infinity norm). 
	

	Bounds on $\infnorm{\err{i}(t)}$ allow us to use various techniques to study the distribution
	of the state at future time steps.
	We illustrate this in Section~\ref{sec:ellipsoid}
  by using tail
  probability approximations  to compute a
  safety area, for which we know that the probability of the system
  landing outside that area is bounded by a predefined constant.
	
	\subsubsection{Global Error Bound}
	First, let
	$
	\xi(t) \eqdefn \displaystyle \max_{0 \leq j \leq j_{0}\dimsys^t}
		\infnorm{\mathbb{E}[x_\ini^{[j]}]}.
	$
	Using $\xi(t)$, we can derive bounds for $\err{j_0}(t)$ by first
	reorganising~\eqref{eq:ej0} according to the moments of $x_\ini$,
	which gives
	\begin{align}
	\err{j_0}(t) = \sum_{j=0}^{j_0\dimsys^t} \widetilde E_j \mathbb E[x_\ini^{[j]}], \label{eq:ej1}
	\end{align}
	where each $\widetilde E_j \matrixdim{n^{j_0}}{n^j}$ is a simple sum
	of products of $E_{n,m}$'s, obtained by this reorganisation
	of~\eqref{eq:ej0}.
	All $\widetilde{E}_j$ can thus be computed
  offline (note that $\widetilde E_j$ is dependent on $t$, but we keep
  this implicit for readability).
	
  From this, we can derive a global bound
  \begin{align} 
    \infnorm{\err{j_0}(t)} \leq \xi(t) \sum_{j=0}^{j_0\dimsys^t} \infnorm{\widetilde{E}_j},
  \end{align}
  where $\sum_{j=0}^{j_0\dimsys^t} \infnorm{\widetilde{E}_j}$ can be computed
  offline.

	Observe that $\xi(t)$
	can be efficiently computed in some cases.
	An obvious situation is when the position $x_\ini$ is determined, in
  which case we have $\xi(t) = \max \{1, \infnorm{x_\ini}^{j_{0}\dimsys^t} \}$.
  Another case is when $x_\ini$ obeys a well-known distribution whose
  moments are easy to compute, such as uniform or normal distributions.
 Another case is when the system satisfies $x(t) \in \mathbb R$ and $x_\ini \in [0, 1]$, in which case 
	$\infnorm{\mathbb{E}[x_\ini^{[j]}]}$ is decreasing and we have 
	$\xi(t) = \infnorm{\mathbb{E}[x_\ini^{[0]}]} = 1$.
	
\subsubsection{
Error Bound using Partial Exact Computation with Block Indices on Moments 
}
We can further refine~\eqref{eq:ej1} to consider a single line $i \leq n^{j_0}$
	of the equation, for which we get
	\begin{align}\label{eq:err_i}
	\err{j_0}_i(t) = \sum_{j=0}^{j_0\dimsys^t} v_{j,i} \mathbb E[x_\ini^{[j]}],
	\end{align}
	where $v_{j, i} \matrixdim{1}{n^j}$ is the $i$th row of $\widetilde{E}_j$. 
	By repeated application of triangle and Cauchy-Schwarz inequalities,
  we have
	\begin{align}\label{eq:err_bound_crude}
	\lvert \err{j_0}_i(t) \rvert
	\leq \xi(t) \sum_{j=0}^{j_0\dimsys^t} \infnorm{v_{j,i}},
	\end{align}
	where $\sum_{j=0}^{j_0\dimsys^t} \infnorm{v_{j,i}}$ can also be
	computed offline.

	This bound can, however, be crude in practice as the norm gets
	distributed over all sums and products. 
	We alleviate this problem in Proposition~\ref{prop:err_bound_J} where we compute tighter bounds while maintaining a reasonable computational cost.
	
	\begin{prop}\label{prop:err_bound_J}
	For any subset $J \subseteq \{0, \ldots, j_0\dimsys^t\}$, we have the bound  
	\begin{align*}
	| \err{j_0}_i(t) | 
	& \leq \Big|\sum_{j\in J}  v_{j,i} \mathbb E[x_\ini^{[j]}] \Big|
	+ \xi(t,J) \sum_{j \not \in J} \infnorm{ v_{j,i} },  
	\end{align*}
	where $\xi(t,J)= \displaystyle \max_{j \not\in J} \infnorm{\mathbb{E}[x_\ini^{[j]}]}$. 
	\end{prop}

	The idea is that $J$ is a set of indices where one should avoid
	distributing the norm over the sum.
	One should pick $J$ to consist of those indices where
	the distribution is too crude and makes the error bound loose.
	In order for this method to be computationally efficient, one should
	pick $J$ that is of relatively small size, e.g., $\card{J} = \OOO(t)$.
	One possible way to choose $J$ is to fix a size $\widehat{j}$ and return the set
	of $\widehat{j}$ indices $j$ where $\norm{v_{j,i}}$ are the largest;
	another way is to return the set of $\widehat{j}$ indices $j$ such that $\infnorm{ \mathbb{E}[x_\ini^{[j]}]}$ are the largest.  
	For example, suppose $x_\ini$ is drawn from a truncated normal distribution over the interval $[0,1]$. 
	Then $\xi(t) = 1$ and $\xi(t,J) = \infnorm{ \mathbb{E}[x_\ini^{[j_{\min}]}] }$, where $j_{\min}$ is the smallest number that is not in $J$. 
	\rv{If $J$ is chosen as the set of first $\widehat{j}$ indices, 
 then $J = \{0, \ldots, \widehat{j}-1\}$ and 
	$\xi(t,J) = \infnorm{ \mathbb{E}[x_\ini^{[\widehat{j}]}] }$.}
  
\subsubsection{\label{subsubsec:boundK} 
Error Bound using Partial Exact Computation with Regular Indices on Moment Coordinates} 
  
We can further refine~\eqref{eq:ej1} and \eqref{eq:err_i} by considering each element of  matrix $\widetilde{E}_j$ and vector $\mathbb{E}[x_\ini^{[j]}]$.
Let 
$\tilde v =  
\begin{bmatrix}
	\widetilde{E}_0 &  \widetilde{E}_1 & \cdots &\widetilde{E}_{j_0 \dimsys^t}
	\end{bmatrix}$ 
and
$\tilde y = \mathbb E
\begin{bmatrix}
	x_\ini^{[0]\intercal} & x_\ini^{[1]\intercal} & \cdots & x_\ini^{[j_0 \dimsys^t]\intercal}
	\end{bmatrix}^\intercal$. 
In the same way as in~\eqref{eq:err_bound_crude}, we have
\[
\rvert \err{j_0}_i(t) \lvert \leq 
\sum_{k=0}^m  \lvert \tilde{v}_{i,k} \rvert \lvert \tilde{y}_k \rvert
\leq 
\max_{k\leq m} \lvert \tilde y_k \rvert \cdot \sum_{k=0}^m  \lvert \tilde{v}_{i,k} \rvert,
\]
where 
$\tilde y_k$ is the $k$th row of $\tilde y$ and
and $\tilde{v}_{i,k}$ is the element at $i$th row and $k$th column of $\tilde v$.
Then, in the same way as in 
 Proposition~\ref{prop:err_bound_J}, we obtain the following proposition.
\begin{prop}\label{prop:err_bound_K}
 For any subset $K \subseteq \{0, \ldots, m\}$, we have 
	\begin{align*}
	| \err{j_0}_i(t) | 
	& \leq \Big|\sum_{k\in K}  \tilde v_{i,k} \tilde y_k \Big|
	+  \max_{k \not\in K} \lvert \tilde y_k \rvert\sum_{k \not \in K} \lvert  \tilde v_{i,k} \rvert.  
	\end{align*}  
\end{prop}
Again, one possible option of the set $K$ is to fix a size $\widehat{k}$ and 
then choose $\widehat{k}$ indices $k$ where $\lvert \tilde y_k \rvert$ are the largest.
In the next section, we will use these error bounds for probabilistic safety analysis.

\section{Ellipsoid Bounds for Probabilistic Safety Analysis}
\label{sec:ellipsoid}

Everything that we have computed up to this point can be computed
offline, as it does not depend on the actual system state.
In this section, we show how to do online computation of
probabilistic safety regions using tail probability analysis.
It crucially relies on approximations of the
first and second moments of the dynamics of the system.

\subsection{Tail Probability Approximation}
\label{subsec:tail_prob_ellipsoid}

In this section, we use the bounds on the error introduced by the
truncated system, which are derived in
Section~\ref{subsec:approx_indiv}, to give a lower bound on the
probability of the system being inside a given ellipsoid region after
$t$ time steps.
\rv{For any $j_0 \leq \dimtru$ and $i \leq n^{j_0}$, let $\Err{j_0}_i(t)$ be an upper bound of $\lvert\err{j_0}_i(t)\rvert$ obtained by any of the methods in  Section~\ref{subsec:approx_indiv}.}

We define the region we are interested in terms of
positive-semidefinite matrices.
A matrix $P \in \real^{n \times n}$ is positive-semidefinite if, for
all vectors $x \in \real^n$, $x^\intercal P x \geq 0$.
Such a $P$ defines a seminorm, called the $P$-seminorm, by $\norm{x}_P
= (x^\intercal P x)^{1/2}$.
The region defined by $\norm{x}_P \leq r$ is an ellipsoid (possibly of
infinite radius in some dimensions).

Recall that $n$
is the dimension of the state vector $x(t)$ of the system~\eqref{eq:system}.
Let 
\begin{align*}
x(t) = \begin{bmatrix}
x_1(t) & x_2(t) & \ldots & x_n(t)
\end{bmatrix}^\intercal \in \real^n,
\end{align*}
where $x_i$'s are scalars.  
Notice that, for any $i, j \in \{1,\ldots,n\}$,
 the $(ni + j)$-th row of $x^{[2]}(t)$ is
\begin{align*}
x^{[2]}_{ni + j}(t) = x_i(t) x_j(t). 
\end{align*} 
Hence, we can approximate the expectations
$\E[x_i(t) x_j(t)]$ using $\truncdim{x}{2}(t)$, which is the approximation of $\E[x^{[2]}(t)]$ computed with the methods in Section~\ref{subsec:approx_indiv}, as well as its error bound $ \Err{2}(t)$.
For the sake of readability, we use the notation
\begin{align*}
\truncdim{x}{2}_{(i,j)}(t) = \truncdim{x}{2}_{ni + j}(t) ,
\end{align*}
as it is the approximation of   
$\E[x_i(t) x_j(t)]$.
Similarly, we use 
\begin{align*}
\Err{2}_{(i,j)}(t) = \Err{2}_{ni + j}(t),
\end{align*}
which is the respective approximation error bound.
We also use $\truncdim{x}{1}_{i}(t)$ and $\Err{2}_{i}(t)$ to denote the $i$-th row of $\truncdim{x}{1}(t)$ and $\Err{1}(t)$, respectively.

Now, we assume that we know a global error bound
\begin{equation}\label{eq epsilon_P}
\norm{\truncdim{x}{1}(t) - \E[x(t)]}_P \leq \ \GErr.
\end{equation} 
We will discuss the technique to get this bound $\GErr$ later at the end of this section.


We then introduce Proposition~\ref{prop:tail_prob_analysis}, which  gives a mean to bound the probability of the system being outside of an ellipsoid centered
around $\truncdim{x}{1}(t)$.


\begin{prop}\label{prop:tail_prob_analysis}
  For given positive-semidefinite matrix $P $ and $\alpha > \GErr$,
  \begin{align}
    \begin{split}
      \proba&(\norm{x(t) - \truncdim{x}{1}(t)}_P \geq \alpha) \\
      &\leq
        \frac{\sum_{i,j=1}^n p_{ij} (\E[x_i(t) x_j(t)] - \E[x_i(t)]
        \E[x_j(t)])}{(\alpha - \GErr)^2}\rlap{,}
      \label{eq:tail_prob_analysis}
    \end{split}
  \end{align}
 where $p_{ij}$ is the element on the $i$th row and the $j$th column of the matrix $P$.
\end{prop}

\begin{proof}
  Since $\alpha > \GErr$, we get by Markov's inequality:
  \begin{align*}
    \proba&(\norm{x(t) - \truncdim{x}{1}(t)}_P \geq \alpha) \\
    & \leq \proba(\norm{x(t) - \E[x(t)]}_P \geq \alpha - \GErr)
      \\
    & = \proba(\norm{x(t) - \E[x(t)]}_P^2 \geq (\alpha -
      \GErr)^2) \\
    & \leq \frac{\E[\norm{x(t) - \E[x(t)]}_P^2]}{(\alpha -
      \GErr)^2} \\
    & = \frac{\sum_{i,j=1}^n p_{ij} (\E[x_i(t) x_j(t)] - \E[x_i(t)]
      \E[x_j(t)])}{(\alpha - \GErr)^2}\rlap{.} &\qedhere
  \end{align*}
\end{proof}


 

For a fixed matrix $P$ and known values of the approximations $\truncdim{x}{1}_i(t)$,
$\truncdim{x}{2}_{(i,j)}(t)$, and the error bounds
$\Err{1}_i(t)$ and $\Err{2}_{(i,j)}(t)$, 
each term of
the sum
in~\eqref{eq:tail_prob_analysis} can be bounded.
The derived bound depends on the signs of $p_{ij}$, $\Err{1}_i(t) -
\truncdim{x}{1}_i(t)$, and $\Err{1}_j(t) - \truncdim{x}{1}_j(t)$.
For example, if $p_{ij} > 0$, $\truncdim{x}{1}_i(t) \geq \Err{1}_i(t)$, and
$\truncdim{x}{1}_j(t) \geq \Err{1}_j(t)$, then we have 
\begin{align}
\begin{split}
  p_{ij}  (\E[&{x}_i(t) {x}_j(t)] - \E[{x}_i(t)] \E[{x}_j(t)]) \leq \\
  & p_{ij} (\truncdim{x}{2}_{(i,j)}(t) + \Err{2}_{(i,j)}(t) - \\
  &\phantom{p_{ij} (} (\truncdim{x}{1}_i(t) -
  \Err{1}_i(t)) (\truncdim{x}{1}_j(t) - \Err{1}_j(t) ))\rlap{.}
\end{split}
  \label{eq: example error bound}
\end{align}
Similar bounds can easily be found in all other cases
by using the following bounds:
\begin{gather*}
	\abs{\E[{x}_i(t) {x}_j(t)] -
		\truncdim{x}{2}_{(i,j)}(t)} \leq \Err{2}_{(i,j)}(t) \rlap{,} \\
	\abs{\E[{x}_i(t)] - \truncdim{x}{1}_i(t)} \leq \Err{1}_i(t)
		\rlap{.}
\end{gather*} 
In particular, if we know exact values for $\E[{x}_i(t)]$ and $\E[{x}_i(t)
{x}_j(t)]$ (i.e., if $\dimtru \geq 2 \dimsys^t$), we have Corollary~\ref{cor:tail_prob_analysis}.
\begin{cor}\label{cor:tail_prob_analysis}
  Consider a positive-semidefinite matrix $P $.
  If $\dimtru \geq 2 \dimsys^t$, for any $\alpha > \GErr$,
  \begin{align}
    \begin{split}
      \proba&(\norm{x(t) - \truncdim{x}{1}(t)}_P \geq \alpha) \\
      &\leq
        \frac{\sum_{i,j=1}^n p_{ij} (\truncdim{x}{2}_{(i,j)}(t) -
        \truncdim{x}{1}_i(t) \truncdim{x}{1}_j(t))}{(\alpha - \GErr)^2}
        \rlap{,}
    \end{split}
  \end{align}
  where $p_{ij}$ is the element on the $i$th row and the $j$th column of the matrix $P$.
\end{cor}


One part that we have left open is how to compute a value for
$\GErr$.
In Section~\ref{subsec:approx_indiv}, we have given bounds on
$\norm{ \mathbb{E}[x(t)] - \truncdim{x}{1}(t)}_\infty$.
This directly gives us bounds on $\norm{  \mathbb{E}[x(t)] - \truncdim{x}{1}(t)}_P$,
using the fact that $\norm{x}_P \leq \lambda_n^{1/2} \norm{x}_\infty$,
where $\lambda_n$ is the greatest eigenvalue of $P$.
{Therefore, $\GErr \leq \lambda_n^{1/2} \infnorm{\err{1}(t)}$. Recall that we can compute an upper bound of $\infnorm{\err{1}(t)}$ using methods in 
 Section~\ref{subsec:approx_indiv}.}

\subsection{Computation of Probabilistic Ellipsoid Bounds}
\label{subsec:compute ellipsiod}
Using the bounds above, we explain how to compute probabilistic
ellipsoid bounds, i.e., ellipsoid areas in which we know the system
will be with 
at least a given  probability.

The problem we are interested in is the following: given a
probabilistic system as in~\eqref{eq:system}, approximations
$\truncdim{x}{1}(t)$ and $\truncdim{x}{2}(t)$ of first and second moments of
the system at time $t$, bounds on the errors of these approximations,
and a constant $b\in(0,1)$, find an ellipsoid in which we know the
system state $x(t)$ will be with probability at least $(1-b)$.
The ellipsoid is preferred to be as small as possible, to give a precise bound.
The problem to find the smallest ellipsoid
 can be formulated as the following optimization problem: 
\begin{align}
\begin{split}
\underset{P}{\text{maximize}} 
	\quad & \det(P)\\
\textrm{subject to} 
	\quad & P  \text{ is positive-definite,}\\
  &\hspace{-4.3em}\rv{\frac{\sum_{i,j=1}^n p_{i,j} (\E[x_i(t) x_j(t)] - \E[x_i(t)]
        \E[x_j(t)])}{(\alpha - \GErr)^2}
      \leq b,} 
      \\
\end{split}\label{eq ellipsoid}
\end{align}  
for some fixed $\alpha > 0$, say $1$,
and $p_{i,j}$ is the element on the $i$th row and the $j$th column of $P$.
\rv{
From 
Proposition~\ref{prop:tail_prob_analysis},
 we have that
the system state at time $t$ is in 
the ellipsoid 
{$\setcomp{x \in \real^{\dimsys}}{\norm{x - \truncdim{x}{1}(t)}_P  
\leq \alpha}$} with probability at least $(1-b)$.}

\begin{remark}\label{remark notconvex}
 The problem~\eqref{eq ellipsoid} cannot be solved by convex optimization
methods -- and therefore not solved online -- because $\GErr$ depends on $P$ and the presence of
$\GErr$ makes it unclear whether the constraint is convex.
Moreover, computing $\GErr$ cannot be done online repeatedly. 
\end{remark}


Hence, we propose a three-step approximate solution for efficiency.
We first solve the problem above, but assume that all errors are $0$ to obtain the matrix $Q$: 
\begin{align}
\begin{split}
\underset{Q}{\text{maximize}} 
	\quad & \det(Q)\\
\textrm{subject to} 
	\quad & Q \text{ is positive-definite,}\\
  & \hspace{-3em} \frac{\sum_{i,j=1}^n q_{i,j} (\truncdim{x}{2}_{(i,j)}(t) -
      \truncdim{x}{1}_i(t) \truncdim{x}{1}_j(t))}{\alpha^2} \leq b\rlap{,}\\
\end{split}\label{eq:simple_ellipsoid}
\end{align} 
where $q_{i,j}$ is the element on the $i$th row and the $j$th column of $Q$.
This makes the problem convex~\cite{boyd2004convex}, so \eqref{eq:simple_ellipsoid} can be solved online.
\rv{Note that changing the value of $\alpha > 0$ does not change the volume of the ellipsoid $\setcomp{x \in \real^{\dimsys}}{\norm{x - \truncdim{x}{1}(t)}_Q  
\leq \alpha}$, since
the volume is inversely proportional to $\det(Q)$. 
So, we may use $\alpha =1$.} 

\rv{
Then, as a second step, we compute a 
matrix $P =s Q$ where $s$ is a positive real value that is obtained by solving the following optimization problem.
\begin{align}
\begin{split}
\underset{s}{\text{maximize}} 
	\quad & \det(s Q)\\
\textrm{subject to} 
	\quad & s \in \mathbb{R}, s > 0, \\ 
      &\hspace{-4.8em} \frac{\sum_{i,j=1}^n s q_{i,j} u_{i,j}(\truncdim{x}{1}(t), \truncdim{x}{2}(t),\Err{1}(t), \Err{2}(t))}{\alpha^2}
      \leq b \rlap{,}\\
\end{split}\label{eq:simple_ellipsoid_adjust}
\end{align}
where $u_{i,j}$ is a function of the moment approximations and truncation error bounds that satisfies
\begin{align*}
\begin{split}
   \E[x_i(t) &x_j(t)] - \E[x_i(t)]
        \E[x_j(t)]\\
    & \leq
    u_{i,j}(\truncdim{x}{1}(t), \truncdim{x}{2}(t), \Err{1}(t), \Err{2}(t)).\label{eq:simple_ellipsoid_adjust_u}
\end{split}
\end{align*}
Notice that $u_{i,j}$ can always be computed as discussed
   in Section~\ref{subsec:tail_prob_ellipsoid} (e.g., \eqref{eq: example error bound}).}
   
\rv{
The goal of this second step is to enlarge the ellipsoid of $Q$ obtained from the first step by using a scalar value $s$. This results in a new ellipsoid $P=sQ$ that satisfies the constraint of \eqref{eq ellipsoid} if $\GErr = 0$.
Note that we could also consider an optimization problem for  $P$ 
that satisfies the constraint of \eqref{eq ellipsoid} (if $\GErr = 0$)
directly, but the problem has too many constraints and is difficult to solve efficiently in practice. 
}
 
Let $P = sQ$ be obtained from solving \eqref{eq:simple_ellipsoid_adjust}.
For the final step, we compute $\GErr$ or its over-approximation 
using the methods in the previous sections.
Recall that, from \eqref{eq epsilon_P}, we have  
$\norm{\truncdim{x}{1}(t) - \E[x(t)]}_P \leq \GErr$.
Hence, we enlarge the ellipsoid of $P$
by using the bound $\GErr$, and take the ellipsoid 
\begin{equation} \label{eq:enlarged_ellipsoid}
\setcomp{x \in
\real^{\dimsys}}{\norm{x  - \truncdim{x}{1}(t)}_{P} \leq \alpha + \GErr}.
\end{equation} 
By 
\eqref{eq epsilon_P},
Proposition~\ref{prop:tail_prob_analysis}, and \eqref{eq:simple_ellipsoid_adjust}, 
the
system state $x(t)$ will be in the ellipsoid~\eqref{eq:enlarged_ellipsoid} with probability at least $(1-b)$.

\begin{figure}[t]
	\centering
		\includegraphics[width=7cm,trim={0cm 0.3cm 0cm 1.49cm},clip]{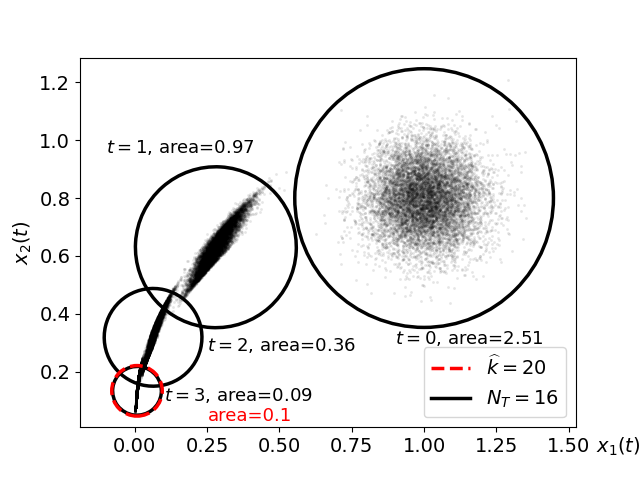}
		\includegraphics[width=7cm,trim={0cm 0.5cm 0cm 1.49cm},clip]{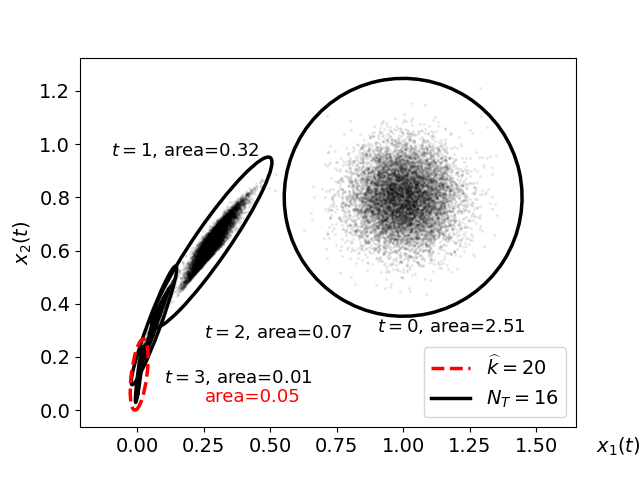}
  \caption{Safe ball (top) and ellipsoid (bottom) bounds for safe regions with
    probability 0.9 for the system in Example~{\ref{ex:4}}. 
    The black bounds are tight bounds computed using $\dimtru = 16$.
    The red dashed bound is computed using an approximate dynamics from Proposition~\ref{prop:err_bound_K} with $\dimtru = 8$ and $\widehat k = 20$.  
    }
  \label{fig ex4}
\end{figure}
\begin{ex}\label{ex:4}
  Figure~\ref{fig ex4} shows the safe regions with probability 0.9 for
  the system in Example~\ref{ex:1} at time steps $t \in \{0,1,2, 3\}$.
	More precisely, we defined $x_1(0)$ (\emph{resp.} $x_2(0)$) as a Gaussian random variable with mean 1 (\emph{resp.} 0.8) and standard deviation 0.1,  
	and $a(t)$ as independent and identically distributed random variables 
	following a uniform distribution on $[0.3,0.4]$.
	The safe regions are plotted in comparison to 10000 Monte Carlo simulations.
  Both plots are computed as described above; 
  however, 
   the regions in the top one are computed
   by restricting to $p_{11} = p_{22}$ and $p_{12} = p_{21} = 0$.
  (i.e., ball-shaped regions).
  The black solid bounds are computed using $\dimtru = 16$; therefore, they are tight bounds where the error $\err{1}(t) =  0$ (see Proposition~\ref{prop:exact-computation}).
	The red dashed bound at $t=3$ is computed using a truncated dynamics and the upper bound of 
the moment approximation errors
 as in Proposition~\ref{prop:err_bound_K}, where  $\dimtru = 8$ and the set $K$ consists of $\widehat k = 20$ indices $k$ where $\lvert \tilde y_{k} \rvert$ are the largest.
	Notice that the bound computed using the truncated dynamics 
 (the red one) is an over-approximation of the tight bound (the black one).
	Notice also that the areas of the ellipsoids are smaller than those of the ball-shaped regions (e.g., at $t = 2$, the area of the ball is 0.36, while that of the ellipsoid is 0.07).
\end{ex}

\begin{remark}
Note that we are in general not interested in optimizing all
dimensions of $P$.
For example, if a system has $(x,y,\theta)$ as coordinates, we may be
only interested in finding a bound for $(x,y)$.
Maximizing {$\det(P)$} under the constraints  in \eqref{eq:simple_ellipsoid} can lead to a
$P$ with a large unit ball in the $\theta$ dimension, but small in the
$x$ and $y$ dimensions, while there may be a better $P$ when
considering just $(x,y)$.
If $I \subseteq \set{1, \ldots ,n}$ is the set of dimensions of
interest, our goal is to maximize {$\det(P_I)$} under the
constraints above, where $P_I$ is the submatrix of $P$ whose indices
are in $I$.
\end{remark}

\section{Smaller matrices with reduced Kronecker powers}
\label{section: reduced Kronecker}
The main bottleneck of our method is the size of the matrix $E$ we compute offline.
One reason for this is due to duplications of computations. Indeed, the Kronecker 
power of a vector contains several times the same element. For example, 
the Kronecker square of the vector 
$x = \begin{bmatrix} a & b \end{bmatrix}^\intercal$ is given by 
$x^{[2]} = \begin{bmatrix} a^2 & ab & ba & b^2 \end{bmatrix}^\intercal$ 
and $ab =ba$ appears twice. In this section, we 
describe a \emph{reduced Kronecker power} whose elements are the same as 
the normal Kronecker power, but without any duplication. For example, the reduced 
Kronecker square of $x$ above will be 
$\redpower{x}{2} = \begin{bmatrix} a^2 & ab & b^2 \end{bmatrix}^\intercal$.

It should be noted that the notion of reduced Kronecker powers was also presented by~\citeasnoun{alma991000034209707066} and~\citeasnoun{carravetta1996polynomial} but used for different problems. 
In this work, we need to efficiently represent those reduced Kronecker
powers, and also accommodate the previous sections with this notion.
This requires the development of an operation corresponding to matrix
multiplication (see~\eqref{eq:reduced_multiplication_ex} below) in
order to propagate moments.
Formally, we need to manipulate polynomials in a clever way, and this
is one of our novelties compared to the literature.

Fix a vector $x = \begin{bmatrix} x_1 & \ldots & x_n \end{bmatrix}^\intercal \vectordim{n}$. 
Each element of a reduced Kronecker power of $x$ 
will correspond to a $n$-tuple of natural number $(m_1, \ldots, m_n)$, this element 
being given by $x_1^{m_1} \ldots x_n^{m_n}$. The degree of such an $n$-tuple is 
given by the sum of its elements $\sum_{i=1}^n m_i$.
Denote the set of $n$-tuples of degree $m$ by $\mathcal{I}_{n,m}$.
This set can be totally ordered by lexicographic order, that is,
$(m_1, \ldots, m_n) < (m'_1, \ldots, m'_n)$ if there is $k$ such that $m_k < m'_k$ and for all $j < k$, $m_j = m'_j$.
For example, $\mathcal{I}_{2,2} = \{(2,0) > (1,1) > (0,2)\}$.

The \emph{$m$-th reduced Kronecker power of $x$} is then given by:
\begin{equation}
\redpower{x}{m} = \begin{bmatrix} x_1^{m_1} \ldots x_n^{m_n} 
	\mid (m_1, \ldots, m_n) \in \mathcal{I}_{n,m} \end{bmatrix}^\intercal.
\end{equation}
This means that the first element of $\redpower{x}{m}$, namely $x_1^m$, 
is given by the largest
element of $\mathcal{I}_{n,m}$, namely $(m, 0, \ldots, 0)$, 
the second element, namely $x_1^{m-1}x_2$, is given by 
the second largest element of $\mathcal{I}_{n,m}$, namely $(m-1, 1, 0, \ldots, 0)$,
and so on. As claimed earlier, the reduced square of 
$x = \begin{bmatrix} a & b \end{bmatrix}^\intercal$ is indeed 
$\redpower{x}{2} = \begin{bmatrix} a^2 & ab & b^2 \end{bmatrix}^\intercal$.

Using this reduced Kronecker power, we can describe our original system in the 
form
\begin{align}
\label{eq:system_reduced}
	\begin{split}
	x(t+1) &= \sum\limits_{i = 0}^{\dimsys} \widehat{F}_i(t)\redpower{x}{i}(t), 
		\quad t\in\nonnegativeInt,\\
	x(0) &=x_\ini.
	\end{split}
\end{align}
Compared to \eqref{eq:system}, where $F_i$ is a matrix of size $n\times n^i$, 
$\widehat{F}_i$ is of size $n\times\lvert\mathcal{I}_{n,i}\rvert$ obtained from 
$F$ by summing columns. For example, if \eqref{eq:system} is of the form:
\begin{equation*}
	\begin{bmatrix} x_1(t+1) \\ x_2(t+1) \end{bmatrix} = 
	\begin{bmatrix} a_{1,1} & a_{1,2} & a_{1,3} & a_{1,4} \\ 
		a_{2,1} & a_{2,2} & a_{2,3} & a_{2,4} \end{bmatrix}
	\begin{bmatrix} x_1(t)^2 \\ x_1(t)x_2(t) \\ x_2(t)x_1(t) \\ x_2(t)^2\end{bmatrix},
\end{equation*}
Then \eqref{eq:system_reduced} is of the form:
\begin{equation}
	\label{eq:reduced_multiplication_ex}
	\begin{bmatrix} x_1(t+1) \\ x_2(t+1) \end{bmatrix} = 
	\begin{bmatrix} a_{1,1} & a_{1,2} + a_{1,3} & a_{1,4} \\ 
		a_{2,1} & a_{2,2} + a_{2,3} & a_{2,4} \end{bmatrix}
	\begin{bmatrix} x_1(t)^2 \\ x_1(t)x_2(t) \\ x_2(t)^2\end{bmatrix}.
\end{equation}
Everything we described in the previous sections can be accommodated with 
this new power, reducing the size of the vectors and the matrices involved. More 
details on the saved space will be given in the experiment section.

In terms of implementation, this power relies on manipulating and generating the 
elements of $\mathcal{I}_{n,m}$ on the lexicographic order. This can be done by representing 
the $n$-tuples as monomials and most calculations can be done using abstract 
polynomials operations. In our implementation, we heavily used the python package
numpoly\footnote{https://pypi.org/project/numpoly/}.

\section{Experimental Results}
\label{sec:numerical-example}
	In this section, we provide two numerical examples and experimental
  results to illustrate our techniques. 
  All experiments are run on a standard laptop computer (MacBook Pro, M1 chip, 16G memory).

	\subsection{Stochastic Logistic Map}\label{subsec:example_slm}
	
	Consider the stochastic logistic map as studied by \citeasnoun{athreya2000},
	which is given by
	\begin{align*}
	x(t+1) &= r(t) x(t) (1-x(t)), \quad t\in \nonnegativeInt, \\
	x(0) &= x_0,
	\end{align*}
	where $x_0, r(0), r(1), \ldots$ are mutually independent random
  variables; $x_0$ takes values in $[0,1]$ and $r(t)$ all take values
  in $[0,4]$.
	The scalar $x(t) \in [0,1]$ represents the population of a species
	subject to growth rate $r(t)$.
	This system can equivalently be represented by~\eqref{eq:system} with
	$\dimsys = 2$, $F_0 (t) = 0$, $F_1 (t) = r(t)$, and $F_2 (t) = -r(t)$.
  For experiments, we chose all $r(t)$ to be uniformly distributed
  over the interval $[0.4, 0.6]$, and $x_0$ to follow a normal
  distribution of mean $0.5$ and standard deviation $0.1$ truncated to
  $[0, 1]$.

\subsubsection{Moment Approximation via Truncated System}

\begin{figure}[!tb]
	\centering
	\includegraphics[width=8cm,trim={0.2cm 0 0.6cm 0.55cm},clip]{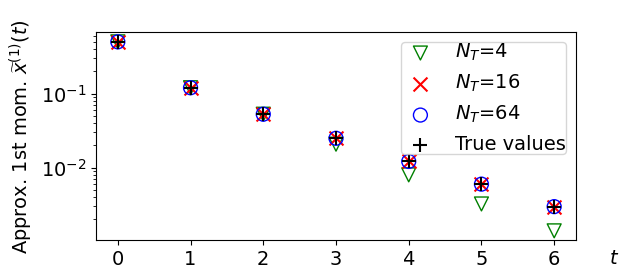}
	\includegraphics[width=8cm,trim={0.2cm 0 0.6cm 0.9cm},clip]{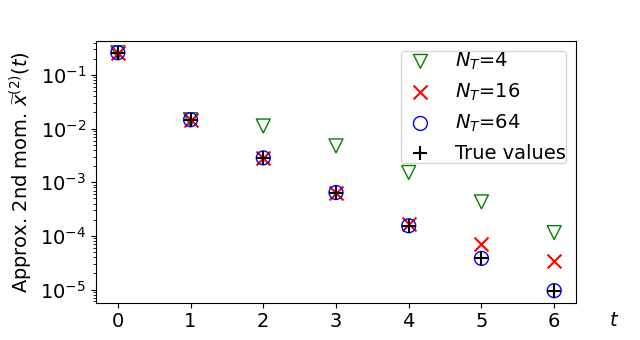}
    \captionspace
    \caption{Moment approximations for stochastic logistic map.}
	\label{fig:moments_logistic_map}
\end{figure}

We first compare our moment approximations for different truncation
limits to the true value of the moments (computed using our method
with $\dimtru = 256$, which gives the true value by Proposition~\ref{prop:exact-computation}).
In Figure~\ref{fig:moments_logistic_map}, we plot the first and second
moments of the truncated system with different truncation limits
$\dimtru$, and larger $\dimtru$ is required to obtain good
approximations of higher moments.
This is a natural consequence, as truncation discards more information on the dynamics of higher moments. 
		
	\subsubsection{Error Bound on Moment Approximations}
	\begin{figure}[!tb]
		\centering
			\includegraphics[width=8cm,trim={0cm 0cm 0cm 0.7cm},clip]{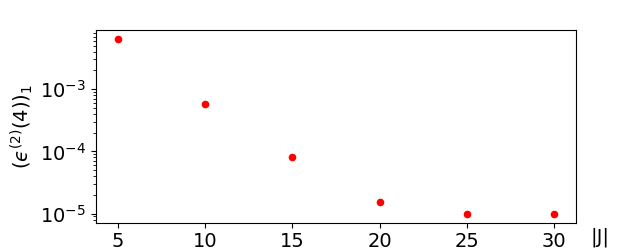}
    \captionspace
    \caption{Error bound on moment approximations.}
    \label{fig:mockup2a}
  \end{figure}
	Next, we evaluate our approximation method of error bounds for moments.
	Figure~\ref{fig:mockup2a} shows error bounds given by 
	Proposition~\ref{prop:err_bound_J}
  with different sizes of $J$, with parameters $\dimtru = 16$, $t = 4$, and $j_0 = 2$.
	The set $J$ contains the indices $j$ where $\lVert \mathbb{E}
  [x_0^{[j]}] \rVert$ are the largest. 
	We observe that the error bound quickly decreases as $|J|$ increases.
  This supports our expectation that we can use our error bounds for
  tail probability analysis with larger parameters ($t$, $\dimtru$,
  and $\dimsys$).
	Note that we cannot expect to get much more precise than the bound for
	$\abs{J} = 30$, since we get the exact error bound for $\abs{J} =
	33$.
	
	\subsubsection{Tail Probability Analysis}
	Lastly, we provide a
	result on tail probability analysis via the method in Section~\ref{sec:ellipsoid}. 
	We computed the error bound for $0 \leq t \leq 5$ using Proposition~\ref{prop:err_bound_J} with $\dimtru
  = 16$ and $|J| = 6 t$, where  $J$ contains the indices $j$ where $\lVert \mathbb{E}
  [x_0^{[j]}] \rVert$ are the largest.
 Figure~\ref{fig:mockup2b} summarizes the results of the analysis.
	Red intervals indicate the $95\%$-probability neighborhoods of
  $\truncdim{x}{1}(t)$ computed by using
  Proposition~\ref{prop:tail_prob_analysis}.
	Blue intervals with a solid line indicate the region where $95\%$ of
  10000 Monte Carlo simulations closest to its mean belong.
	Dotted intervals indicate the range of 10000 Monte Carlo simulations.
	We observe that the size of safety intervals given by our tail probability analysis is reasonably small. 
	It becomes cruder in later time steps. 
	This is expected, as the approximation error of moments, which is a bottleneck in refining the error bounds, becomes larger as time progresses (cf.\ approximate 2nd moment in Fig.~\ref{fig:moments_logistic_map}).
	
	There are two major advantages of our method compared to Monte Carlo simulation. 
	One is that our technique computes moment approximations much faster 
 (even for large $\dimtru$ 
	and a small number of samples) because we do not rely on generating 
	random numbers.
	This advantage is highlighted in Table~\ref{table:time}, which
  contains the online computation times for Monte Carlo simulations
  and our approach, averaged over 100 runs.
  The offline computation of our approach takes $0.014$ seconds for $\dimtru = 256$.
	Another advantage is that our safety interval gives a theoretical guarantee on probabilistic safety that cannot be achieved by Monte Carlo simulations. 

  \begin{table}
		\mytablecaption{Comparison of online computation times.}{table:time}
    \scriptsize
    {
    \begin{center}
      \begin{tabular}{|c|c|c|c|c|c|c|}
        \hline
        Method & \multicolumn{2}{c|}{Monte Carlo} & \multicolumn{4}{c|}{Moment propagation} \\
        \hline
        \multirow{2}{*}{Parameters} & \multicolumn{2}{c|}{num.\ samples} &
          \multicolumn{4}{c|}{$\dimtru$} \\
          \cline{2-7}
          & $10$ & 
          $10^4$ & $4$ & $16$ & $64$ & $256$ \\
        \hline
        Time ($\mu$s) & $4.4$e$10^3$ & 
          $2.8$e$10^5$ &
          49  & 51 & 57 & 67 \\
        \hline
      \end{tabular}
    \end{center}
    }
  \end{table}

%

	\begin{figure}[!tb]
		\centering
			\includegraphics[width=7.5cm,trim={0.1cm 0 0.5cm 0.9cm},clip]{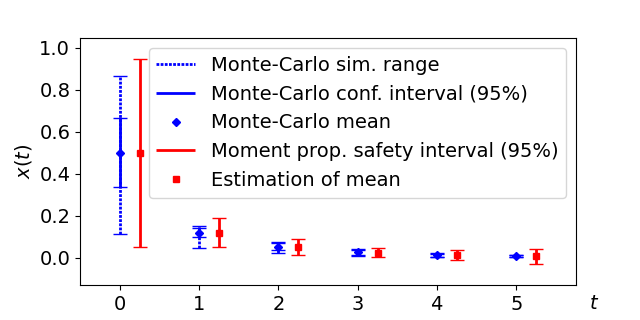}
        \captionspace
		\caption{Tail probability analysis with Monte-Carlo simulation and moment propagation with truncation limit $\dimtru = 16$.}
		\label{fig:mockup2b}
	\end{figure}

	\subsection{Application to Automated Driving}	\label{subsec:example_vd}
  Our second example is an application to automated driving.  For
  safety guarantees, autonomous vehicles need to predict their
  future positions.  One way to achieve this is set-based
  reachability, as advocated by \citeasnoun{althoff2014online}.
  To use their method,
  they must consider systems with bounded disturbances and use
  linearization around an equilibrium, that is, they approximate a polynomial 
  system by a linear one, using Lagrange remainders.
  Our method is based on Carleman linearization, which allows taking the effect of higher dimensions
  of the system into account more precisely than Lagrange remainders.
  Moreover, our approach is probabilistic, while theirs is set-based,
  so the two approaches give different types of guarantees.
	
	We consider a scenario in which, at each time step, the 
	vehicle measures its current position with some known sensor error
  distributions, computes moments of its current state, 
  and predicts its
  future positions 
  up to $t$ steps ahead in time by
  applying the truncated system to these moments.
	\subsubsection{Vehicle Dynamics}
	
	More precisely, we consider the \emph{kinematic bicycle model} of a
	vehicle from~\citeasnoun{kong2015kinematic}, which we rewrite as the following
	equivalent polynomial system
	\begin{gather*}
	\dotcoordx(t) = v(t)c(t), \quad
	\dotcoordy(t) = v(t)s(t), \\
	\dot \psi(t) = \frac{v(t)}{\ell}\sin\beta, \quad
  \dot v(t) = a(t), \\
	\dot c(t) = -\frac{s(t)v(t)\sin\beta}{\ell}, \quad
	\dot s(t) = \frac{c(t)v(t)\sin\beta}{\ell},
	\end{gather*}
	for $t \geq 0$, where $\coordx (t)\in \mathbb{R}$ and $\coordy (t)\in\mathbb{R}$
	represent the X--Y coordinates of the mass-center of the vehicle,
	$v(t)\in\mathbb{R}$ denotes its speed, $\psi(t)$ its inertial
  heading, and $a(t)\in\mathbb{R}$ its acceleration.
	The constants $\beta\in\mathbb{R}$ and $\ell>0$ respectively denote
	the angle of velocity and the distance from the vehicle's rear axle to
	its mass-center.
	The scalars $c(t)$ and $s(t)$ are auxiliary variables that are
	introduced to obtain the polynomial model above from the original
	model of \citeasnoun{kong2015kinematic} (which involves
	trigonometric terms), using the same techniques
	as \citeasnoun{carothers2005some}.
	
  The second-order Taylor expansion of the model above gives the
  following discrete-time approximation:
	\begin{align*}
	\coordx(t+\Delta) &= \coordx(t) + \Delta c(t)v(t) \nonumber \\
	&\quad+ \frac{\Delta^2}{2}\biggl(a(t)c(t) - \frac{s(t)v^2(t)\sin\beta}{\ell}\biggr), \\
	\coordy(t+\Delta) &= \coordy(t) + \Delta s(t)v(t) \nonumber \\
	&\quad+ \frac{\Delta^2}{2}\biggl(a(t)s(t) + \frac{c(t)v^2(t)\sin\beta}{\ell}\biggr), \\
	\psi(t+\Delta) &= \psi(t) + \Delta\frac{v(t)}{\ell}\sin\beta + \frac{\Delta^2}{2}\frac{a(t)}{\ell}\sin\beta,\\
	v(t+\Delta) &= v(t) + \Delta a(t), \\
	c(t+\Delta) &= c(t) - \Delta\frac{s(t)v(t)\sin\beta}{\ell} \nonumber \\
	&\quad- \frac{\Delta^2}{2}\biggl(\frac{c(t)v^2(t)\sin^2\beta}{\ell^2} + \frac{a(t)s(t)\sin\beta}{\ell}\biggr),\\
	s(t+\Delta) &= s(t) + \Delta\frac{c(t)v(t)\sin\beta}{\ell} \nonumber \\
	&\quad+ \frac{\Delta^2}{2}\biggl(-\frac{s(t)v^2(t)\sin^2\beta}{\ell^2} + \frac{a(t)c(t)\sin\beta}{\ell}\biggr),
	\end{align*}
	where $\Delta>0$. To describe the evolution of the vehicle states at times $0, \Delta, 2\Delta, \ldots$, we write this system in the form of~\eqref{eq:system}. In particular, consider the discrete-time instant $t \in \nonnegativeInt$ corresponding to the continuous time $t\Delta$. By letting
	\begin{align*}
	x(t) &\eqdefn [\coordx(t), \coordy(t), \psi(t), v(t), c(t), s(t)]^{\intercal},
	\end{align*}
	we obtain \eqref{eq:system} with $\dimsys = 3$ and the coefficients
  $F_0(t),\ldots,F_3(t)$ depend on $\Delta, \beta, \ell$, and $a(t)$.
  We consider the setting where the acceleration values $a(0), a(1), \ldots$ are
  independent uniformly-distributed random variables over $[0.9,1]$, $\Delta =
  0.1$, $\beta = \pi / 8$, $\ell = 2.5$, and, for the initial state, $\coordx(0)$,
  $\coordy(0)$, $v(0)$, $\psi(0)$ are independent Gaussian random
  variables with mean $0$ and standard deviation $0.1$, and $c(0) =
  \cos(\psi(0) + \beta)$ and $s(0) = \sin(\psi(0) + \beta)$.


  \subsubsection{Experimental Results}

	\begin{figure}[!tb]
		\centering
			\includegraphics[width=7.5cm,trim={0.5cm 0 0.5cm 0.9cm},clip]{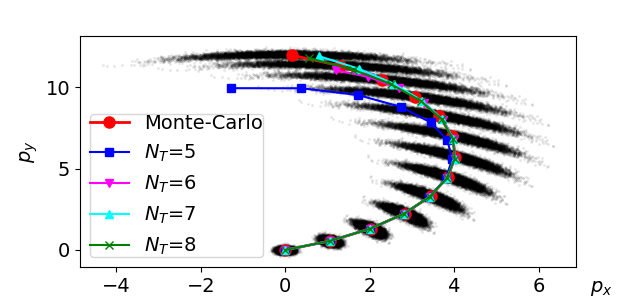}
        \captionspace
		\caption{First moment approximation of vehicle dynamics.}
		\label{fig:vehicle}
	\end{figure}
	\begin{figure}[!tb]
		\centering
			\includegraphics[width=8cm,trim={0cm 0.2 0.1cm 0.9cm},clip]{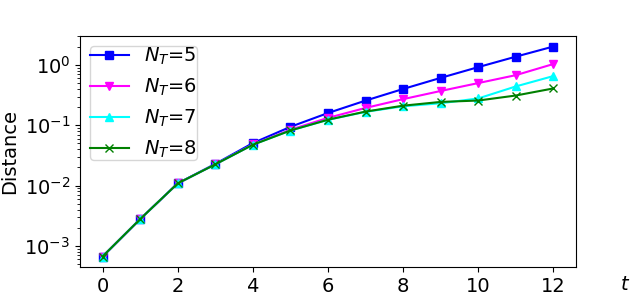}
    \captionspace
    \caption{Distance to the mean of the empirical distribution.}
		\label{fig:vehicle_distance}
	\end{figure}
	\begin{figure}[!tb]
		\centering
			\includegraphics[width=8cm,trim={0.2cm 0.2cm 0.1cm 0.9cm},clip]{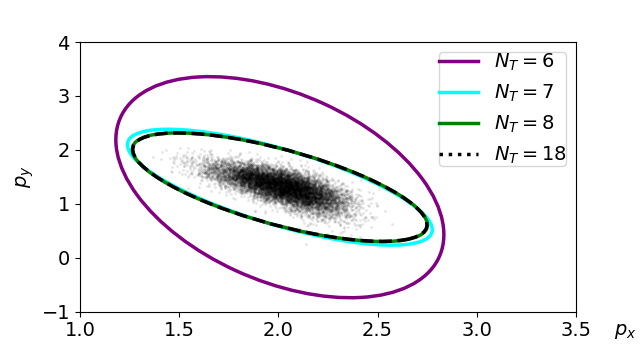}
    \captionspace
    \caption{The ellipsoids of the safe region for time step $2$ and probability bound 0.9 computed 
    using~\eqref{eq:enlarged_ellipsoid}
    and Proposition~\ref{prop:err_bound_K} 
    with $\widehat k = 900$ and different truncation limits $\dimtru$.    
 Using $\dimtru = 18 $, we obtain the region computed using the exact first and second moments (see Proposition~\ref{prop:exact-computation}).
    }
		\label{fig:vehicle_ellipsoids_NTs}
	\end{figure}
	\begin{figure}[!tb]
		\centering
			\includegraphics[width=8cm,trim={0.2cm 0.5 0.2cm 0.5cm},clip]{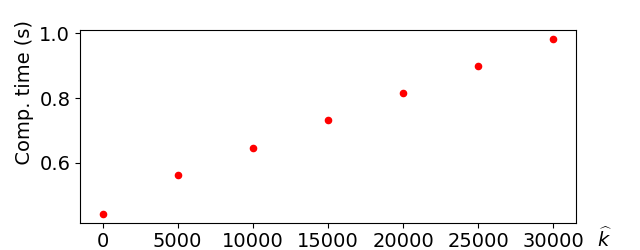}
    \captionspace
    \caption{The computation time for the error bounds on  the approximations for $\E[p_x]$, $\E[p_y]$, $\E[p_x^2]$, $\E[p_y^2]$, and $\E[p_x p_y]$, at $t=2$,  using Proposition~\ref{prop:err_bound_K} with $\dimtru = 8$ and different sets $K$, where each of the set $K$ contains $\widehat k$ indices $k$ where $\lvert \tilde y_{k} \rvert$ are the largest.
    } \label{fig:vehicle_error_comp_time}
	\end{figure}
	\begin{figure}[!tb]
		\centering
			\includegraphics[width=8.2cm,trim={0.1cm 0.2cm  0.2cm 0.9cm},clip]{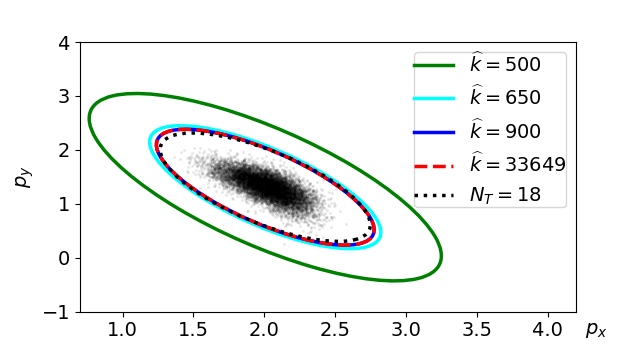}
    \captionspace
    \caption{The ellipsoids of the safe region for time step $2$ and probability bound 0.9 computed 
     using~\eqref{eq:enlarged_ellipsoid}
    and Proposition~\ref{prop:err_bound_K} with $\dimtru = 7$ and different sets $K$. Each ellipsoid is computed using the set $K$ consisting of $\widehat k$ indices $k$ where $\lvert \tilde y_{k} \rvert$ are the largest. Using $\dimtru = 18 $, we obtain the region computed using the exact moments (see Proposition~\ref{prop:exact-computation}).
    } \label{fig:vehicle_ellipsoids}
	\end{figure}
		
  Figure~\ref{fig:vehicle} shows the expected trajectory of the
  vehicle as approximated by our method for different truncation
  limits, as well as the empirical distribution computed by 10000 runs
  of Monte Carlo simulation and the mean of that
  distribution.
  Figure~\ref{fig:vehicle_distance} shows the distance
   between $\truncdim{x}{1}(t)$ and the mean of
  the empirical distribution for the same truncation limits.
  Figures~\ref{fig:vehicle} and~\ref{fig:vehicle_distance} show that
  larger truncation limits give truncated systems that follow the empirical
  distribution closer.
  It also shows that, for a fixed truncation limit, the distance to
  the empirical system grows larger with time. 

  Figure~\ref{fig:vehicle_ellipsoids_NTs} shows the ellipsoids of the safe region at time step 2 with probability bound 0.9 computed using~\eqref{eq:enlarged_ellipsoid} 
  and Proposition~\ref{prop:err_bound_K} 
  with different truncation limits $\dimtru$. 
  We can see that a larger $\dimtru$ gives a more precise  bound.

Figure~\ref{fig:vehicle_error_comp_time} shows the computation time (averaged over
1000 runs) for the error bounds on moment approximation at time step $2$ computed using Proposition~\ref{prop:err_bound_K}  with different sizes of the sets $K$.  
This result shows that, as expected, the computation time increases linearly with the size of $K$.

  Figure~\ref{fig:vehicle_ellipsoids} also shows the ellipsoids of the safe region at time step 2 with probability bound 0.9, but with different index sets $K$. 
  We can see that the larger $K$ is, 
  the closer the ellipsoids get to that obtained with the exact computation using $\dimtru = 18$ and Proposition~\ref{prop:exact-computation} (the black dashed bound). 
 Note that the set $K$ containing all 33,649 indices means that it contains all moments needed for an exact computation of the upper bound of truncation error at $t = 2$ 
 using~\eqref{eq:ej0} and the reduced Kronecker powers in Section~\ref{section: reduced Kronecker}. 
 In other words, the red dotted bound is the tightest bound we can obtain using our method for the truncation limit $\dimtru = 7$. We note that the bound computed with   $\widehat k = 900$ is very close to this tightest bound and their computation is much faster.
  

  \subsection{Comparison of Kronecker Powers and Reduced Kronecker Powers}
  
  	In Section~\ref{subsec:example_vd}, we used the
	reduced Kronecker powers described in 
	Section~\ref{section: reduced Kronecker} to generate the matrices 
	$E(\dimtru,\dimtru)$. Here, we illustrate the gain, both in time 
	and space, obtained by using reduced Kronecker powers instead 
	of non-reduced ones. The comparison results are compiled in 
	Figure~\ref{fig:comparison_kron_poly}.
	
	\begin{figure}[!tb]
		\centering
			\includegraphics 
            [width=7.5cm,trim={0cm 0.3 0.1cm 0.9cm},clip]{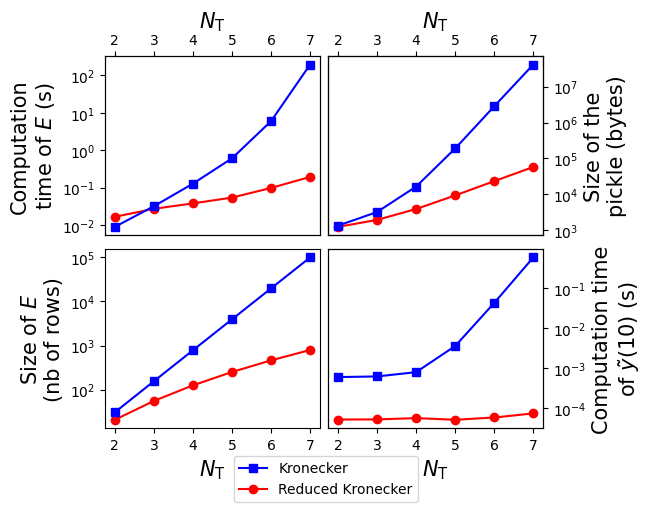}
        \captionspace
		\caption{Comparison between Kronecker powers and reduced ones}
		\label{fig:comparison_kron_poly}
	\end{figure}
	
	We compared the following performance indicators, for $\dimtru \in \{2, \ldots, 7\}$:
	\begin{itemize}
		\item The time to compute $E(\dimtru,\dimtru)$. 
		\item The size of the compressed data file (pickled npz format in Numpy/Python) containing $E(\dimtru,\dimtru)$.		 
		\item The number of rows of $E(\dimtru,\dimtru)$. 
		\item The time required to compute $\tilde{y}(10)$. 
	\end{itemize}
	We remark that using reduced Kronecker powers significantly 
	improves the performance in all aspects, making both the online and the offline computations 
	much faster, and the memory load much lighter.
	Furthermore, we also note that the computation of $E(8,8)$ using non-reduced 
	Kronecker powers reached an out-of-memory error after several 
	hours of computation, while we could compute $E(25,25)$ in less 
	than an hour with reduced ones.
  

\section{Conclusions}\label{sec:Conclusion}
In this paper, we have proposed a method to approximate the moments of a discrete-time stochastic polynomial system. This method is built upon a Carleman linearization approach with truncation, where approximate moments are obtained by propagating initial moments through a finite-dimensional linear deterministic difference equation. We have presented guaranteed bounds on the approximation errors. We have then used the approximate moments and the approximation error bounds together with a convex optimization technique to provide probabilistic safety analysis. We have demonstrated our method on a stochastic logistic map and a vehicle model with stochastic acceleration inputs.

Our moment approximation method is applicable to systems with both additive and multiplicative noise. Furthermore, it provides probabilistic guarantees even when the dynamics and the noise probability distributions are complicated. 
Our method involves computations in two phases: an initial offline computation phase where the approximate moment dynamics are obtained, and an online computation phase that involves propagation of the initial moments through the obtained dynamics. The online phase is very fast and we have shown in our numerical examples that our method can provide moment approximations in much shorter times compared to the Monte Carlo approach based on repeated simulations. For the offline computations, we have investigated a technique to improve the efficiency by using the symmetry of Kronecker powers and reducing the sizes of the matrices involved in the computations.

While in this paper we have addressed only polynomial systems, our method can also be applied to certain nonpolynomial systems. In some cases, nonpolynomial dynamics can be transformed to polynomial ones by introducing auxiliary variables, as illustrated in one of our numerical examples. In other cases, polynomial approximations can be useful.

We have used the approximated moments of the system's state for computing tail-probability bounds of the state being outside of ellipsoidal safety regions. As pointed out by  \citeasnoun{schmudgen2017moment} and \citeasnoun{john2007techniques}, approximate moments of random variables can be useful for approximating probability distributions. One of our future research directions is to investigate approximation of the probability distribution of the system state 
using its approximated moments.
\bibliographystyle{ifac}

\appendix
\section{Proof of Theorem~\ref{theorem H}}
\label{appendixA}
First,
we have the simple observation in Lemma~\ref{lem:Hjk_rec}.
\begin{lem}
	\label{lem:Hjk_rec}
	The sets $H_{j,k}$ satisfy the 
	 recursive formula 
	\begin{equation*}
	H_{j+1,k} = \bigcup_{n=0}^{\dimsys} \set{n} \times H_{j,k-n}\rlap{,}
	\end{equation*} 
	with the convention that $H_{j,k} = \emptyset$ for $k < 0$.
\end{lem} 
Next, we show in Lemma~\ref{lemHd2} that Theorem~\ref{theorem H} holds for $j=2$. 
We omit $t$ for simplicity. 
\\ 
\begin{lem}\label{lemHd2}
$
\left(F_0 x^{[0]} + F_1   x^{[1]}+ \ldots + F_{\dimsys} x^{[\dimsys]}\right) ^ {[2]}  = 
\displaystyle\sum_{k=0}^{2\dimsys} \left(
      \smashoperator[r]{\sum_{(i_1, i_2)\in H_{2,k}}}
      F_{i_1} \otimes  F_{i_2} \right) x^{[k]}.
$
\end{lem}  
\begin{proof}
Recall that Kronecker product has the 
\emph{mixed-product property}  
	(i.e., $(A \otimes B)(C \otimes D) = AC \otimes BD$), but not \emph{commutative property} ($(A \otimes B) = (B \otimes A)$ may not hold) (see Section~13.2 of \citeasnoun{laub2005matrix})).
	Then, this lemma holds by Equation~\eqref{proof:A1}.
\end{proof}  
Now, we are ready to prove Theorem~\ref{theorem H}. 
\begin{proof}[Proof of Theorem~\ref{theorem H}]
The theorem obviously holds for $j = 0$ and $j=1$, and also holds for $j = 2$ by Lemma~\ref{lemHd2}. We will show the case where $j>2$ by induction.
We assume that the theorem hold for $j = n$, and consider the
inductive case where $j = n+1$ in Equation~\eqref{proof:3}, where we write
$\bigotimes_{m = 1}^n F_{i_m}$ for $F_{i_1} \otimes \ldots \otimes
F_{i_m}$ for brevity.
Note that the last line holds by adding terms that correspond to the
same moment $x^{[k]}$ (such terms span diagonals in the large sum
in Equation~\eqref{proof:3}) and by Lemma~\ref{lem:Hjk_rec}.
Therefore, the theorem holds for $j = n+1$, which concludes the proof.
\end{proof}    

\onecolumn

\begin{align}\label{proof:A1}\begin{split}
&\left(F_0  x^{[0]} + F_1 x^{[1]} + \ldots  + F_{\dimsys} x^{[\dimsys]}\right) ^ {[2]}  
\\
&= F_0  x^{[0]} \otimes \left(F_0  x^{[0]}  + \ldots  + F_{\dimsys}x^{[\dimsys]}\right)  
+ F_1 x^{[1]}\otimes\left(F_0  x^{[0]} + \ldots  + F_{\dimsys} x^{[\dimsys]}\right)+\ldots+ F_{\dimsys} x^{[\dimsys]}\otimes \left(F_0  x^{[0]}  + \ldots + F_{\dimsys}x^{[\dimsys]}\right)
\\
&= (F_0 \otimes F_0) x^{[0+0]}   + (F_0 \otimes F_1) x^{[0+1]} 
 + (F_0 \otimes F_2) x^{[0+2]}   +\ldots + (F_0 \otimes F_{\dimsys}) x^{[0+{\dimsys}]} 
\\
&  \hspace{2.8cm} +(F_1 \otimes F_0) x^{[1+0]}  + (F_1 \otimes F_1) x^{[1+1]} +\ldots + (F_1 \otimes F_{\dimsys -1}) x^{[1+{(\dimsys-1)}]} + (F_1 \otimes F_{\dimsys}) x^{[{1+\dimsys}]} 
\\
& \hspace{5.6cm} + (F_2 \otimes F_0) x^{[2+0]} +\ldots  
\\
&\hspace{8.3cm}+ \ldots+(F_{\dimsys}  \otimes F_{0}) x^{[{\dimsys}+0]} + 
\qquad 
\ldots +(F_{\dimsys} \otimes F_{\dimsys}) x^{[{\dimsys+\dimsys}]}  
\\
&=  \smashoperator[r]{\sum_{(i_1, i_2) \in H_{2,0}}}
       {\big(}F_{i_1} \otimes F_{i_2}  \big) x^{[0]}  + 
    \smashoperator[r]{\sum_{(i_1, i_2) \in H_{2,1}}}
       {\big(}F_{i_1} \otimes F_{i_2}  \big) x^{[1]}  + \ldots +
   \smashoperator[r]{\sum_{(i_1, i_2) \in H_{2,2 \dimsys}}}
      {\big(}F_{i_1} \otimes F_{i_2}  \big) x^{[2 \dimsys]}  
= 
\sum_{k=0}^{2\dimsys} \left(
      \smashoperator[r]{\sum_{(i_1, i_2)\in H_{2,k}}}
      F_{i_1} \otimes  F_{i_2} \right) x^{[k]}.
\end{split}
\end{align} 
\begin{align}\label{proof:3}
\begin{split}
&(F_0  x^{[0]} + F_1 x^{[1]} + \ldots  + F_{\dimsys} x^{[\dimsys]}) ^ {[n+1]}  
\\
&= (F_0  x^{[0]} + F_1 x^{[1]} + \ldots  + F_{\dimsys} x^{[\dimsys]}) \otimes (F_0  x^{[0]} + F_1 x^{[1]} + \ldots  + F_{\dimsys} x^{[\dimsys]}) ^ {[n]}  
\\
&= (F_0  x^{[0]} + F_1 x^{[1]} + \ldots  + F_{\dimsys} x^{[\dimsys]}) \otimes  
\sum_{k=0}^{n\dimsys} \left(
      \smashoperator[r]{\sum_{(i_1, \ldots, i_n)\in H_{n,k}}}
      F_{i_1} \otimes \ldots \otimes  F_{i_n} \right) x^{[k]}
\\
&= F_0  x^{[0]} \otimes 
\sum_{k=0}^{n\dimsys} \left(
      \smashoperator[r]{\sum_{(i_1, \ldots, i_n)\in H_{n,k}}}
      F_{i_1} \otimes \ldots \otimes  F_{i_n} \right) x^{[k]}
+ \ldots 
+ F_{\dimsys}  x^{[{\dimsys}]} \otimes 
\sum_{k=0}^{n\dimsys} \left(
      \smashoperator[r]{\sum_{(i_1, \ldots, i_n)\in H_{n,k}}}
      F_{i_1} \otimes \ldots \otimes  F_{i_n} \right) x^{[k]}
\\
&=\quad
      \smashoperator[r]{\sum_{(i_1, \ldots, i_n)\in H_{n,0}}}
      \big( F_0 \otimes   \bigotimes_{m = 1}^n F_{i_m}\big) x^{[0+0]} 
    + \smashoperator[r]{\sum_{(i_1, \ldots, i_n)\in H_{n,1}}}
      \big( F_0 \otimes  \bigotimes_{m = 1}^n F_{i_m}\big)  x^{[0+1]} 
    + \ldots + 
    \smashoperator[r]{\sum_{(i_1, \ldots, i_n)\in H_{n,n \dimsys}}}
      \big( F_0 \otimes  \bigotimes_{m = 1}^n F_{i_m}\big)  x^{[0+n \dimsys]} 
\\
&\quad 
	+\smashoperator[r]{\sum_{(i_1, \ldots, i_n)\in H_{n,0}}}
     \big( F_1 \otimes  \bigotimes_{m = 1}^n F_{i_m} \big) x^{[1+0]} 
    +\smashoperator[r]{\sum_{(i_1, \ldots, i_n)\in H_{n,1}}}
     \big( F_1 \otimes  \bigotimes_{m = 1}^n F_{i_m} \big) x^{[1+1]}  
    +   \ldots 
    + \smashoperator[r]{\sum_{(i_1, \ldots, i_n)\in H_{n,n \dimsys}}}
     \big( F_1 \otimes  \bigotimes_{m = 1}^n F_{i_m} \big) x^{[1+n \dimsys]}       
\\
&\quad +\ldots
\\
&\quad 
	+\smashoperator[r]{\sum_{(i_1, \ldots, i_n)\in H_{n,0}}}
     \big( F_{\dimsys} \otimes  \bigotimes_{m = 1}^n F_{i_m} \big) x^{[ \dimsys + 0]} 
    +\smashoperator[r]{\sum_{(i_1, \ldots, i_n)\in H_{n,1}}}
     \big( F_{\dimsys} \otimes  \bigotimes_{m = 1}^n F_{i_m} \big) x^{[ \dimsys+1]}  
    +   \ldots 
    + \smashoperator[r]{\sum_{(i_1, \ldots, i_n)\in H_{n,n \dimsys}}}
     \big( F_{\dimsys} \otimes  \bigotimes_{m = 1}^n F_{i_m} \big) x^{[ \dimsys +  n\dimsys]}  
\\
 &=\sum_{k=0}^{(n+1)\dimsys} \left(
      \smashoperator[r]{\sum_{(i_1, \ldots, i_{n+1})\in H_{n+1,k}}}
      F_{i_1} \otimes \ldots \otimes  F_{i_{n+1}} \right) x^{[k]}. 
\end{split}
\end{align}

\end{document}